\documentclass{sig-alternate}
\usepackage{hyperref}

\usepackage{amsmath,amssymb,amsthm,amsfonts,latexsym}        
\usepackage{url}
\usepackage{graphicx,subfigure}
\usepackage{bm}
\usepackage[show]{chato-notes} 
\usepackage{paralist}
\usepackage{color}
\usepackage[normalem]{ulem}

\newcommand{\eat}[1]{} 
\newcommand{\favg}{\mbox {$f_{avg}$}}  
\newcommand{\cavg}{\mbox {$c_{avg}$}}  
\newcommand{\gell}{\mbox {$G_{\ell}$}}  
\newcommand{\im}{\mbox {\sc MaxSpread}}  
\newcommand{\mintss}{\mbox {\sc MinTss}}  
\newcommand{\calc}{\mbox {${\cal C}$}} 
\newcommand{\pioptk}{\mbox {$\pi_{opt, k}$}} 
\newcommand{\piopt}{\mbox {$\pi_{opt}$}} 

\newcommand{\SV}{\textcolor{black}}

\newcommand{\C}{\eat}

\begin{document}	
\title{ Adaptive Influence Maximization in Social Networks: Why Commit when You can Adapt? }
\numberofauthors{2} 
\author{
\alignauthor
Sharan Vaswani \\
       \affaddr{University of British Columbia}\\
       \affaddr{Vancouver, Canada}\\
       \email{sharanv@cs.ubc.ca}      
\alignauthor
Laks V.S. Lakshmanan \\
       \affaddr{University of British Columbia}\\
       \affaddr{Vancouver, Canada}\\
       \email{laks@cs.ubc.ca}
}

\maketitle

\newtheorem{theorem}{Theorem} 
\newtheorem{fact}{Fact} 
\newtheorem{example}{Example} 
\newtheorem{problem}{Problem} 
\newtheorem{prop}{Proposition} 

\begin{abstract}
Most previous work on influence maximization in social networks is limited to the non-adaptive setting in which the marketer is supposed to select all of the seed users, to give free samples or discounts to, up front. A disadvantage of this setting is that the marketer is forced to select all the seeds based solely on a diffusion model. If some of the selected seeds do not perform well, there is no opportunity to course-correct. A more practical setting is the adaptive setting in which the marketer initially selects a batch of users and observes how well seeding those users leads to a diffusion of product adoptions. Based on this market feedback, she formulates a policy for choosing the remaining seeds. In this paper, we study adaptive offline strategies for two problems: (a) \im\ -- given a budget on number of seeds and a time horizon, maximize the spread of influence and (b) \mintss\ -- given a time horizon and an expected number of target users to be influenced, minimize the number of seeds that will be required. In particular, we present theoretical bounds and empirical results for an adaptive strategy and quantify its practical benefit over the non-adaptive strategy. We evaluate adaptive and non-adaptive policies on three real data sets. We conclude that while benefit of going adaptive for the \im\ problem is modest,  adaptive policies lead to significant savings for the \mintss\ problem. 


\eat{Most of the previous literature on influence maximization in social networks is concerned with the non-adaptive setting in which the marketer is supposed to select all of the seeds initially. A more natural variant of this problem is the adaptive setting in which the marketer gives free products / discounts to some of the users initially, then checks how the market reaction and depending on the market feedback formulates a strategy to seed other nodes. In this paper, we study adaptive strategies for two problems - a) given a budget, find the nodes to be seeded so as to maximize the spread of influence in the network and b) minimize the number of initial seeds to influence a certain number of users in the network. 
In particular, we present theoretical bounds and empirical results for adaptive seeding and quantify its benefit over the non-adaptive strategy. Given a fixed time horizon, we find offline adaptive policies for the two problems described above. } 
\end{abstract}

\section{Introduction}
\label{sec:Introduction}
\label{sec:intro} 
Recently, there has been tremendous interest in the study of influence propagation in social and information networks, motivated by applications such as the study of spread of infections and innovations, viral marketing, and feed ranking to name a few (e.g., see \cite{hethcote2000mathematics,domingos2001mining,samper2008nectarss,song2007information}). A prototypical problem that has received much attention in the literature is  \emph{influence maximization} (\im): given a directed network $G$, with edge weights denoting probabilities of influence between nodes, find $k$ nodes, such that activating them in the beginning leads to the maximum  \emph{expected spread}, i.e., expected number of activated nodes as according to a given diffusion model. For the viral marketing application, nodes may model users, activation may correspond to product adoption, and seed users are given free or price discounted samples of the product, with the aim of achieving the maximum expected number of product adoptions. 

Kempe, Kleinberg and Tard\"{o}s \cite{kempe2003maximizing} formalized this as a discrete optimization problem  and  studied several discrete-time diffusion models 
including independent cascade and linear threshold (details in Section~\ref{sec:Related-Work}). They showed that \im\ under these models is NP-hard but the expected spread function satisfies the nice properties of monotonicity and submodularity. Exploiting these properties, they showed that a simple greedy algorithm, which repeatedly adds the seed with the largest marginal gain, i.e., increase in expected spread, until the budget $k$ is reached, achieves a $(1-1/e)$-approximation to the optimum. There has been an explosion of research activity around this problem, including development of scalable heuristics, alternative diffusion models, and scalable approximation algorithms (e.g., see \cite{chen2009efficient} \cite{wang2012scalable} \cite{leskovec2007cost} \cite{goyal2011simpath} \cite{goyal2011data} \cite{Tang2014Influence}). \SV{For space limitations, we refer the reader to \cite{chen2013information} for a more detailed survey.}

\C{Monotonicity means adding more seeds cannot lead to a lower expected spread whereas submodularity intuitively corresponds to the property of diminishing returns (formal definitions in Section~\ref{sec:Related-Work}). By exploiting these properties and early results by Nemhauser et al.\cite{nemhauser1978analysis}, they showed that a simple greedy algorithm which repeatedly selects the node that results in the maximum marginal gain in expected spread provides a $(1-1/e)$-approximation to the optimal solution.} 

The majority of the work in influence maximization has confined itself to a \emph{non-adaptive} setting where, in viral marketing terms, the marketer must commit to choosing all the $k$ seeds up front. This means that the choice of every single seed is driven completely by the diffusion model used for capturing the propagation phenomena. In practice, it may happen that the \emph{actual spread} resulting from the seeds chosen may fall short of the expected spread predicted by the diffusion model. Recent work by Goyal et al. \cite{goyal2011data} shows that most diffusion models tend to over-predict the actual spread. Thus, committing to the choice of all $k$ seeds in one shot can result in a sub-optimal performance in actuality. A more realistic setting is one where the marketer chooses a subset of seeds and activates them. She monitors how their activation spreads through the network and observes the actual spread thus far. She can then take into account this market feedback in making subsequent seed selections. We call this setting an \emph{adaptive} setting, as choices of subsequent seeds are adapted to observations made so far about the actual spread achieved by previous selections. Hence, the adaptive setting introduces a \emph{policy} $\pi$ which specifies which node(s) to seed at a  given time. It is very intuitive that adaptive seed selection should lead to a higher actual spread compared to non-adaptive seed selection, since it benefits from market feedback and tailors seed selections accordingly. 

\C{As explained, the network is represented as a probabilistic graph with the edge weights denoting the probabilities that a user will influence its neighbour.} 
\C{In the IC model, each active user gets one chance to influence her  neighbour. We refer to this as an activation attempt. An activation attempt succeeds with the corresponding edge probability. An edge along which an activation attempt succeeded  is said to be ``live'' whereas the other edges are said to be ``dead''. Since there are $|E|$ edges in the network and each one of them can be live or dead, there are $2^{|E|}$ possibilities. These are referred to as possible worlds or possible realizations of the network. In the real world, the status of all these edges will be uniquely determined resulting in a deterministic network. Such a deterministic network which gives information about the reality (how the information diffused from a set of initially active nodes). This deterministic network is one among the set of possible worlds. Since it reflects the reality, we call it the ``true'' possible world.}
Adaptive seed selection raises several major challenges. For instance, in the adaptive setting, in practice, there is a finite time horizon $T$ within which the marketer wishes to conduct her viral marketing campaign. Suppose $k$ is  the seed budget of a marketer. The marketer must then consider the following questions. How many seeds to select at a given time, that is, what is the \emph{batch} size? Which nodes should be selected in each intervention ? How long should she wait between seeding successive batches (interventions)? If $T$ is sufficiently long, it seems intuitive that selecting one seed at a time and waiting until the diffusion completes, before choosing the next seed,  should lead to the maximum spread.
The reason is that we do not commit any more seeds than necessary to start a fresh diffusion and every seed selection takes full advantage of market feedback. 
We refer to the above case as \emph{unbounded time horizon}. Another natural question is, what if the time horizon $T$ is not long enough to allow many small batches to be chosen and/or diffusions to be observed in full. In this case, which we call \emph{bounded time horizon}, the marketer has to choose a strategy in which the budget $k$ is spent within the time horizon $T$ and every seed selection benefits from as much feedback as possible. 

Instead of maximizing the spread, the marketer may have a certain expected spread as the \emph{target} that she wants to achieve. This target may be derived from the desired sales volume for the product. A natural problem is to find the minimum number of seeds needed to achieve the target. This problem, called \emph{minimum targeted seed selection} (MINTSS for short), has been studied in the non-adaptive setting \cite{goyal2010approximation}, where it was shown that the classic greedy algorithm leads to a bi-criteria approximation to the optimal solution.\C{Precisely, if the target spread is $Q$ and a shortfall of $\beta > 0$ is allowed, then the greedy algorithm will achieve an expected spread $\ge (Q-\beta)$ using no more than $OPT (1+\ln\lceil Q/\beta \rceil)$ seeds, where $OPT$ is the optimal number of seeds.} An interesting question is whether an adaptive strategy for seed selection can significantly cut down on the number of seeds needed to reach a given target spread. 

\C{
Adaptive influence maximization has been proposed previously in~\cite{guillory2010interactive} and~\cite{golovin2011adaptive}. While ~\cite{guillory2010interactive} is more directed towards active learning, the authors propose an idea to maximize influence in a social network with hidden information. In~\cite{golovin2011adaptive}, the authors introduce the notion of ``adaptive submodularity'' and derive average case bounds for the performance of greedy adaptive policies for the sensor selection problem, the active learning setting, and viral marketing. Finally, ~\cite{chen2013near} addresses the MINTSS problem described above and shows that under certain conditions, the batch-greedy adaptive policy (in which the seeds are chosen in batches in a greedy manner) is competitive not only against the sequential greedy policy (choosing one seed at a time) but also against the optimal adaptive policy. A detailed review of related work appears in Section ~\ref{sec:Related-Work}.
}

Adaptive \im\ has been studied recently in \cite{guillory2010interactive,  golovin2011adaptive, guillory2010interactive}, and adaptive MINTSS has been studied in \cite{chen2013near}. While a more detailed comparison with these papers appears in Section~\ref{sec:Related-Work}, here are the key differences with our work. 
The market feedback model assumed by these papers is that when a node is activated (seeded), a subset of the out-edges from the node become active or ``live'' while others stay inactive or become ``dead''. This amounts to saying we get to observe which active nodes succeeded in activating which other nodes. We refer to this as ``edge level'' feedback. \emph{Edge level feedback assumption is unrealistic}, since in practice, we may only know which other nodes activated as a result of choosing certain seeds, rather than who succeeded in activating whom. 

The experiments conducted in these papers (if at all) are on small toy networks with $1000$ nodes and they do not clarify the {\sl practical} benefits of going adaptive for real large networks. All previous studies are confined to the setting of \emph{unbounded time horizon}, which means the horizon is long enough for the diffusion started by each batch to complete. In practice, the horizon may be bounded and not leave enough time for successive diffusions to complete. The theoretical results in these papers bound the performance of the greedy adaptive policy compared to the optimal adaptive policy. Notice that the optimal (adaptive) policy cannot be computed in polynomial time. The only practical options for both non-adaptive and adaptive settings are greedy approximations (possibly with techniques for scaling up to large datasets). Thus, a real question of practical interest is what do we gain by going adaptive, i.e., what is the gain in peformance of the greedy approximation algorithm when it is made adaptive? In contrast, we study\im\ and MINTSS under both unbounded and bounded time horizon and quantify the benefits of going adaptive with reference to the greedy approximation algorithm, as opposed to the optimal algorithm which is not practical. Furthermore, we propose a novel node level feedback model which we use for adaptive seed selection. Node level feedback is in line what is really observable in practice: {\sl which users became active as a result of seeding the last batch of users?}  

\eat{Also, all of the previous work addresses the problem in the unrealistic case of an unbounded time horizon. They derive relations between the greedy adaptive policy and the unimplementable optimal policy. To the best of our knowledge, we are the first to characterize the practical benefit of going adaptive for both of the problems described above by comparing the greedy non-adaptive and greedy adaptive policies. We conduct extensive experiments on large real world datasets and consider the problem of adaptive influence maximization for both unbounded and bounded time horizon. } 

In this paper, we address the aforementioned questions and make the following contributions. 
\begin{compactitem} 
\item We define the problems of \emph{adaptive influence maximization} under bounded or unbounded time horizon and \emph{minimum adaptive targeted seed selection} by generalizing their non-adaptive counterparts (Section~\ref{sec:Problem-Definition}). 

\item We propose a novel \emph{node level feedback} model for capturing market feedback for adaptive seed selection and show that as long as the time horizon is unbounded, i.e., long enough to allow diffusions to complete, the spread function for our node level feedback model is adaptive monotone and adaptive submodular (Section~\ref{sec:Theory}). 

\item We establish a bound on the spread achieved by a greedy adaptive strategy for seed selection compared to both an optimal adaptive strategy and a greedy non-adaptive strategy. The former shows that the greedy algorithm continues to provide a guaranteed approximation, while the latter formally establishes the benefits of practical adaptive strategies over practical non-adaptive ones (Section~\ref{sec:Theory}). 

\item We establish a similar bound on the number of seeds required by the greedy adaptive policy compared to the greedy non-adaptive one, in order to meet a given target expected spread and establish the practical advantage of going adaptive (Section~\ref{sec:Theory}).  

\item For the unbounded horizon, we scale up the classic greedy adaptive policy by leveraging the recent state-of-the-art non-adaptive randomized algorithm based on reverse reachable sets and adapt it to the adaptive setting to achieve superior performance (Section~\ref{sec:Algorithms}). 

\item \C{Since the greedy algorithm for adaptive seed selection for the unbounded horizon is sub-optimal when the horizon is bounded,} We argue that the expected spread function is computationally hard to optimize f	or the case of bounded time horizon and propose an alternative algorithm based on sequential model-based optimization (Section~\ref{sec:Algorithms}).  

\item We conduct a comprehensive set of experiments on 3 real datasets to measure the performance of our algorithms and their advantages over non-adaptive policies. We report our results (Section~\ref{sec:Experiments}).  
\end{compactitem}

Related work is discussed in Section~\ref{sec:Related-Work}. We summarize the paper and present directions for future work in Section~\ref{sec:Conclusion}. 

\eat{ 
Let G represent a social network with nodes representing users and the edges connections between them. Weights for edges correspond to influence probabilities which correspond to the influence of users on each other. Information spreads in the network according to a diffusion model. The diffusion models are inspired from sociology and propagation of virus propagation. The most commonly used discrete diffusion models (in which time unfolds in discrete steps) are the Independent Cascade (IC) and the Linear Threshold (LT) model. In this paper, we consider only the IC model of propagation. Most of our results can be generalized to the LT model trivially.In the IC model, each node becomes active either because it was seeded initially at time $t = 0$ or become active in the previous time-step because of an active neighbour. Each active node gets one chance to activate its inactive neighbours. The probability of a successful activation is proportional to the influence probability between the two users. An activation attempt by a user along an edge is probabilistic and each edge in the network can be in a binary state indicating whether an activation attempt will be successful or not, we obtain an exponentially large number of possible worlds. The `true' world which reflects the activation attempts in reality is a realization of the one of the possible worlds. 

Given a budget $k$ on the number of nodes that can be seeded initially, the objective of influence maximization is to find the $k$ nodes which when made active initially will maximize the expected number of nodes becoming activated under the given diffusion model. Since we do not know the true world beforehand, we maximize the expected number of active nodes where the expectation is across the exponentially many possible worlds. ~\cite{kempe2003maximizing} formulated the \im\ problem as a combinatorial optimization problem. They prove that under certain diffusion models the expected spread (number of active nodes at the end of the diffusion process) follows the property of submodularity of law of diminishing returns. They exploit this submodularity property and use the results in ~\cite{nemhauser1978analysis} to propose a greedy algorithm. The greedy algorithm at each stage calculates the marginal gains (the expected increase in spread because of seeding that particular node) for each of the nodes in the network and adds to the seed set the node with the highest marginal gain. Since calculating the marginal gain involves a computation of over a large number of possible worlds, we resort to Markov Chain Monte-Carlo (MCMC) simulations to give an estimate of the marginal gains. It can be shown that the greedy algorithm obtains at least a $1 - e - \epsilon$ approximation of the optimal spread. Here $\epsilon$ refers to the error incurred in computing the marginal gains. Since marginal gain computation is expensive and we need to calculate the expected marginal gains of all the nodes in the graph for choosing the next seed, the greedy algorithm is not scalable. In particular, the time complexity for influence maximization is of the order $\Omega(kmn poly(\epsilon^{-1}))$ where $n$ is the number of nodes and $m$ is the number of edges in the network. Most innovation in influence maximization ~\cite{leskovec2007cost}~\cite{goyal2011simpath}~\cite{chen2010scalable} over the recent year involves making the influence maximization procedure more scalable. A recent paper ~\cite{Tang2014Influence} proposes a randomized algorithm which does away with the use of MCMC simulations and has an expected time complexity of $O((k+l)(n+m)logn / \epsilon^{2})$ and returns a $(1- e - \epsilon)$ approximate solution with $1 - n^{-l}$ probability. This algorithm is around 4 orders of magnitude faster than the previous state of the art methods for influence maximization. 

A more natural variant of this problem is the adaptive setting in which the marketer gives free products / discounts to some of the users initially, then checks the market reaction and depending on the market feedback formulates a strategy to seed other nodes. For example, the marketer may seed the first $k/2$ nodes and wait for a certain time to observe how the information has spread in the real world. Based on the observation about the market reaction, the marketer will choose the remaining seeds to maximize the spread. Thus the marketer is able to incorporate his observations about the market in selecting the subsequent seeds and hence we expect an adaptive strategy will lead to a better spread in the true world. This can also be seen from the point of view of possible worlds. The marketer cares about maximizing the influence / spread in the (unknown) true world. The initial seeds are chosen to maximize the expected number of active nodes where the expectation is over the exponentially many possible worlds. Because of averaging, the seed set can lead to an arbitrarily bad spread in the true world. In the above example, when the marketer chooses the initial seeds and observes the diffusion in the network, he discovers some part of the true possible world. This eliminates some of the possible worlds which do not satisfy the discovered structure about the network. Thus while selecting the remaining seeds, the averaging needs to be done across the possible worlds which are closer (in the structural sense) to the true possible world. This will lead to a seed set which will lead to better spread in the true world. This provides intuition on why an adaptive strategy will be better than a non-adaptive strategy. 

Usually the marketer cares about maximizing the influence in the network within a certain time period. This introduces the notion of time horizon $T$ which is the maximum time to which a diffusion process can be run. This leads to a natural extension of the above problem namely given a budget of size $k$, find the initial seed nodes which will maximize the expected number of activated users within the time horizon. Since we are concerned with only a discrete diffusion process in which one hop corresponds to one timestep, the maximum time a diffusion process can run for is equal to the length of the longest path in the graph. In this paper, we consider two variants of the influence maximization problem - the infinite time horizon case in which we don't care about when a user gets influenced and the finite time horizon case which more closely resembles reality. Most social networks are small world networks ~\cite{newman2003structure} and have a small longest path. Hence, for our analysis we assume that the time horizon is more than the length of the longest simple path of the network. For the non-adaptive policy, in which all the seeds are selected at time $t = 0$, the assumption on the time horizon implies that influence maximization for the finite time horizon is the same as that for that infinite case. For the adaptive policy however, a bounded time horizon imposes restrictions on the policies which will be able to maximize the influence within a given time. We formalize this problem in later sections. 

A marketer may have a certain target of users he wants to influence within some time. This problem involves finding the smallest (or more generally minimum cost) set of nodes which when activated initially will achieve a certain target spread $Q$. This problem has been addressed in the non-adaptive context by Goyal.et.al ~\cite{goyal2010approximation} who refer to this problem as MINTSS and prove bounds for a bi-criteria approximation. We consider the adaptive variant for this problem as well. This problem has also been addressed in the adaptive setting in ~\cite{golovin2011adaptive}~\cite{guillory2010interactive}~\cite{chen2013near} all of which derive bounds for adaptive policies for MINTSS problem. For this problem too, we explore both the infinite time horizon and the bounded time horizon cases. 

The paper is structured as follows: Section~\ref{sec:Related-Work} describes the related literature on adaptive policies for influence maximization. Section~\ref{sec:Problem-Definition} mathematically formulates the problem we wish to solve. We state our major theoretical results in section~\ref{sec:Theory} and present our algorithms in section~\ref{sec:Algorithms}. We describe the experimental setup and our results in section~\ref{sec:Experiments}. Finally,  we conclude the paper in section~\ref{sec:Conclusion} and outline interesting directions for future work in section~\ref{sec:Future-Work}. 
}

\section{Related Work}
\label{sec:Related-Work}
\SV{
\eat{In this section, we present a brief summary of non-adaptive influence maximization and mention the work on adaptive influence maximization in detail. As mentioned in section~\ref{sec:Introduction}, the objective of viral marketing is to maximize the expected number of users who adopt the product(expected spread) given a budget constraint on the number of users to which the marketer can give free products / discounts.} 
\noindent 
{\bf Non-adaptive}: Two classical models of diffusion that have been extensively studied are Independent Cascade (IC) and the Linear Threshold (LT) \cite{kempe2003maximizing}. Both are discrete-time models. 
\eat{In the IC model, the focus of this paper, each active node at time $t$ has one indepenent chance to activate each of its out-neighbors. An attempt succeeds with probability given by the corresponding edge weight. If the attempt succeeds, the neighbor will become active at time $t+1$.}
 The expected spread function under both these models is monotone and submodular. A real-valued set function $f:2^U\rightarrow R$ is \emph{monotone} if $f(S) \le f(S'), \forall S\subset S'\subseteq U$. It is \emph{submodular} if $\forall S\subset S' \subset U$ and $x\in U\setminus S' $, $f(S' \cup \{x\}) - f(S') \le f(S\cup\{x\}) - f(S)$, i.e., the marginal gain (increase in the objective function) by adding an element to a set cannot increase as the set grows. 
\eat{This is also known as the diminishing returns property.} 
While \im\ under both IC and LT models is NP-hard, a simple greedy algorithm \cite{nemhauser1978analysis} provides a $(1-1/e)$-approximation to the optimal solution.}
\C{The expected spread function $\sigma$ in a social network, under popular models such as IC and LT, is submodular \cite{kempe2003maximizing}. In viral marketing, we would like to select a seed set of a certain size $k$ which maximizes $\sigma$.} 
The greedy algorithm involves successively selecting the node with the highest marginal gain. 
\C{Nemahuaser et al.  \cite{nemhauser1978analysis} showed that for monotone submodular functions, the greedy algorithm achieves a $(1 - 1/e)$-approximation of the optimal value.} 
Since computing the expected spread of a given set (and hence marginal gain) is \#P-hard for both IC and LT models \cite{weic-icdm2010, chen2010scalable}, Kempe et al.  \cite{kempe2003maximizing} advocated using MCMC simulations to estimate marginal gains. Using MCMC estimation of the marginal gain, the greedy algorithm yields a $(1-1/e-\epsilon)$-approximation to the optimum, where $\epsilon>0$ is the error in marginal gain estimation. 
Tang et al. \cite{Tang2014Influence} propose a near-optimal (w.r.t. time complexity) randomized greedy $(1-1/e-\epsilon)$-approximation algorithm for \im. It uses the concept of random reverse reachable (RR) sets to achieve this. We briefly review their procedure in Section~\ref{sec:Algorithms}. Here, we note that it is currently the state of the art for \im\ and has been shown to scale to a billion-edge network  \cite{Tang2014Influence}. 

\noindent 
{\bf Adaptive}: Adaptive influence maximization has been proposed previously in ~\cite{golovin2011adaptive, guillory2010interactive, chen2013near}.
Golovin and Krause~\cite{golovin2011adaptive} extend the definitions of submodularity and monotonicity to the adaptive setting. In the adaptive setting, batches of nodes are seeded at different intervals. When a batch is seeded, an \emph{actual} diffusion (called realization in \cite{golovin2011adaptive}) unfolds as per the classical IC model. The next batch is chosen based on the previously observed cascade. An objective function is \emph{adaptive monotone} and \emph{adaptive submodular} if the marginal gain of every element is non-negative and non-increasing in every possible realization, as the size of the set (alternatively length of the policy) increases. We wish to choose a policy that maximizes such an objective function in the adaptive setting. As before, the greedy policy consists of selecting the node with the maximum marginal gain. Golovin and Krause  \cite{golovin2011adaptive} derive average case bounds on the performance of greedy adaptive policies. They also prove bounds on the greedy adaptive policy for adaptive submodular functions under matroid constraints~\cite{golovin2011adaptive2}. 
They assume an edge level feedback mechanism with the IC  model and show that the expected spread is adaptive monotone and adaptive submodular, guaranteeing an approximation algorithm. 
Guillory et al.~\cite{guillory2010interactive} study the problem of submodular set cover in the adaptive setting in which the objective is to minimize the total number of sets required to cover a certain target set and prove worst case bounds for the greedy adaptive policy. 
They briefly describe how their framework can be used for influence maximization in a social network with hidden information (e.g., hidden preferences of  users). 
In this paper, we consider the more traditional influence maximization problem and assume that users do not have any hidden preferences. We establish average case guarantees similar to~\cite{golovin2011adaptive}.
Finally, ~\cite{chen2013near} addresses the adaptive \mintss\ problem and shows that under certain conditions, the batch-greedy adaptive policy, in which the seeds are chosen in batches in a greedy manner, is competitive not only against the sequential greedy policy (choosing one seed at a time) but also against the optimal adaptive policy. 
As explained in the introduction, the key difference between these papers and our work is that unlike them, we adopt a more realistic node level feedback, establish bounds relating adaptive greedy policy with the non-adaptive greedy algorithm, thus answering the question, in practice what does one gain by going adaptive. This question is not answered by simply comparing adaptive greedy with the optimal policy which anyway could not be used in a real network owing to its intractability.

\C{As mentioned in the introduction, unlike previous work, we characterize the \emph{practical} benefit of going adaptive by comparing greedy non-adaptive and greedy adaptive policies. We study adaptive influence maximization under both unbounded and bounded time horizon. Finally, our experiments are run on much larger datasets compared to previous work on adaptive influence maximization.}

\section{Problem Definition}
\label{sec:Problem-Definition}
We consider a directed social network $G = (V, E)$, with $|V| = n$ and $|E| = m$ with the edge weights giving the influence probabilities between two users. In this paper, we assume the independent cascade (IC) diffusion model. In this model, time proceeds in discrete steps. At time $t=0$, the seed nodes are active. Each active user gets one chance to influence/activate her neighbor in the next time step. This activation attempt succeeds with the corresponding influence probability between the two users. An edge along which an activation attempt succeeded  is said to be ``live'' whereas the other edges are said to be ``dead''. 
This leads to $2^{|E|}$ \emph{possible worlds} 
of the network. One of these possible worlds is the \emph{true} world, which reflects the reality of which activation attempts succeeded/failed and of the whole cascade/diffusion that occurred starting from the batches of seeds chosen at various times. Clearly, we don't know the true world, although as described below, in an adaptive setting, it may be revealed partially from successive seed selections. 

\eat{In the ''real'' world, the status of all these edges will be uniquely determined resulting in a deterministic network. Such a deterministic network, which gives information about the reality (how the information diffused from a set of initially active nodes), is one among the exponentially many possible worlds. Since it reflects  reality, we call it the \emph{true}  
world.} 

We consider influence maximization in the adaptive setting, where we are given a seed budget $k$ and a time horizon $T$. The \emph{state} of a network at time $t$ is used to determine the seed(s) which will be selected at that time. The precise definition of state depends on the type of feedback model. Golovin and Krause \cite{golovin2011adaptive} consider an edge level feedback, where they assume after a batch is seeded, the precise status (live/dead?) of every edge is observable. We will instead adopt the \emph{node level} feedback model, whereby we assume the status (active/inactive?) of every node in the network is observable. 
\eat{Under edge level feedback, we know whether edges incident on current active nodes are live/dead, where a node is active iff it's reachable from seeds via live edges \cite{golovin2011adaptive}.}  
A \emph{network state} is a mapping $\psi_t:V\rightarrow \{0,1\}$. For convenience, $\psi_t$ can be treated as the set of active nodes at time $t$. At $t=0$, the nodes seeded at $t = 0$ are the only active nodes. Let $\gell$ denote the true world. It is easy to observe which nodes become active: e.g., which users buy a product or share a particular story on their Facebook page. These are precisely the nodes reachable from seeds in $\gell$ within a given \emph{time horizon} $T$. However, we cannot observe which edges of the network $G$ are actually present (i.e., are live) in $\gell$: e.g., we can't trace which friend of a particular user influenced her to buy the product. Hence node level feedback is a more realistic feedback model. Notice that it makes weaker assumptions about what is observable, than edge level feedback. 

\eat{ 
\SV{In the next section, we prove that 
the two models are equivalent from the perspective of the marginal gains computed and the subsequent seeds selected.}} 

By a \emph{policy}, we mean a seed selection strategy, i.e., a mapping $\pi: \psi_t \rightarrow 2^V$, where $\psi_t$ is the network state at current time $t$. Thus, a policy specifies which nodes to seed next, given the current network state. It is possible to have $\pi(\psi_t) = \{\}$. This means the policy chooses to not seed any nodes and just ``wait''. The policy thus transforms a network from its current state $\psi_t$ to the new state $\psi_{t+1} := \psi_t \cup \pi(\psi_t)$. The network state can change from seeding as well as from the spreading of influence from current active nodes under the diffusion model. We use $f(\pi)$ to denote the spread of a policy $\pi$ within a given time horizon $T$ in the true world (the world $\gell$ which is realized in reality). The seeds selected by an adaptive policy depend on the true world whereas for the non-adaptive case, the selected seeds are independent of the true world. We care about the performance of both strategies in the true world. Since {\sl we} don't know the true world, we quantify the \emph{average} gain obtained by going adaptive, and do so by generating a number of candidate true worlds, finding the spread $f(\pi)$ in each of these worlds and taking the average. We denote the average performance of a policy $\pi$ as $\favg(\pi)$. $\sigma$ refers to the expected spread without any feedback from the network whereas $f_{avg}$ is the average true spread given a feedback model. Note that $\sigma$ is equal to $f_{avg}$ for a non-adaptive policy since it does not use any feedback. 

We use $\pi_k$ to denote a policy constrained to select $k$ seeds. When the policy $\pi$ seeds a node, i.e., when $\pi(\psi_t)\ne\emptyset$, we refer to it as an \emph{intervention}. Since the diffusion model we  use is discrete, the diffusion process completes in a maximum of $D$ time-steps where $D$ is the length of the longest simple path in the network. If 
the last intervention occurred at time $t$ and there are no subsequent interventions, then we have $\psi_{t + D + 1} = \psi_{t + D}$. 

Since we have a model of the diffusion process, we can simulate reality (i.e., generate or sample a set of candidate true worlds) and gauge policies by measuring their performance against the set of candidate true worlds. Thus, policy design is based on maximizing the average performance over the set of candidate worlds. These worlds constitute a training set in standard machine learning. The true world (the test sample in this analogy) will be generated from the same diffusion process, i.e., it will be drawn from the same probability distribution. Since the policy is known before we actually implement it in the real world, these are \emph{offline} policies, as opposed to \emph{online} policies, which are determined dynamically depending on the state of the network and the time left. 
In this paper, we focus on offline policies. 

\C{Suppose a policy performs an intervention at time $t$. We now describe the feedback obtained from the network after the intervention. In the edge level feedback model considered in \cite{golovin2011adaptive}, we can observe the status (live/dead) of every edge exiting $v$ for all nodes $v$ which are reachable via live edges in the true world from the initially activated nodes. Live edges in the true world $\gell$ correspond  to the edges along which activation attempts succeeded in reality. 
}

\eat{In the simplest most intuitive case, the marketer can observe which users bought the item / shared the article on their facebook timeline etc. Thus he can observe all of the active nodes in the network. We call this feedback model node level feedback. Golovin and Krause in~\cite{golovin2011adaptive} describe a stronger feedback model which they call Full adoption feedback. The feedback model states that we know the status (live/dead) of every edge exiting $v$ for all nodes $v$ which are reachable via live-edges in the true world from the initially activated nodes. Live edges in the true world refer to the edges along which activation attempts succeeded in reality or alternatively the edges which are present in the deterministic network representing the true world. 
} 

\eat{ If the diffusion process is allowed to complete, the nodes reachable (via live edges) from the initially activated seed nodes correspond to the set of active nodes. Hence edge level feedback implies that we know the status of all the edges exiting every active node. Edge level feedback is impractical since a marketer cannot observe the state of all edges in the entire network, she can only see which users shared a story or bought a product. 
}

We next formalize the two problems studied in this paper. We assume $\calc$ is a set of candidate true worlds, chosen by some oracle, which will be used for computing $\favg$.

\begin{problem}[\im] 
\label{prob:prob1} 
Given a directed probabilistic network $G = (V,E)$, seed budget $k$, and time horizon $T$, find the optimal policy $\pioptk = arg\,max \{\favg(\pi_k)\}$. 
\end{problem} 
\vspace{-0.1 cm}
Notice that in the calculation of spread $f$ (and hence $\favg$), the diffusion is restricted to a length of at most $T$, where $T$ is the given time horizon. When $T \ge kD$, we say the time horizon is \emph{unbounded}. Otherwise, it is \emph{bounded}. We consider \im\ under both bounded and unbounded horizons. 

\eat{ 
For the influence maximization problem, let $k$ be the initial budget. The conventional problem can be stated as: find the set $S$ of $k$ nodes which when activated initially will maximize the $\sigma$ under an appropriate diffusion model. Formally,
\begin{eqnarray}
maximize \hspace{0.25cm} \sigma(S) \nonumber \\ 
s.t. \hspace{0.2 cm} | S | \leq k
\label{eq:spread-problem}
\end{eqnarray}
For the adaptive case, the objective is to find an optimal policy $\pi_{k}$ which will maximize $f_{avg}$. 
} 

\mintss\ in the non-adaptive case aims to find a seed set of the smallest size such that the resulting expected spread is above a given threshold $Q$. The best known result states there is a bi-criteria approximation: given $\beta > 0$, the greedy algorithm yields a seed set of size $\ge OPT(1 + \ln{\lceil\frac{Q}{\beta}\rceil})$, where $OPT$ is the optimal seed set size and the expected spread of the greedy selection is $\ge Q-\beta$. We generalize it to the adaptive setting. Let $\cavg(\pi)$ denote the average number of seeds chosen by policy $\pi$ across the true worlds $\calc$ on which it is tested. 

\begin{problem}[\mintss] 
\label{prob:prob2}
Given a directed probabilistic network $G = (V,E)$, time horizon $T$, and spread threshold $Q$, find the policy $\piopt$ that leads to the smallest seed set, i.e., $\piopt = arg\,min \{\cavg(\pi) \mid f(\pi)\ge Q\}$. 
\end{problem}
\vspace{-0.1 cm}
\eat{ 
As defined in Section~\ref{sec:Introduction}, \mintss\  seeks to minimize the size of the seed set which when activated will lead to an estimated target spread of $Q$. Intuitively we need to find which users to give free products to, to achieve the target spread and minimize the cost of the advertising campaign. To make the problem less stringent, we allow for a shortfall of $\beta$ i.e a policy / set if feasible if it achieves a spread of at least $Q - \beta$. We can formally define the non-adaptive version of the problem as: find the set $S$ such that:
\begin{eqnarray}
minimize \hspace{0.25cm} | S | \nonumber \\
s.t. \hspace{0.2 cm} \sigma(S) \geq Q - \beta
\label{eq:cover-problem}
\end{eqnarray}
} 

Notice that policy $\pi$ may end up using different numbers of seeds in different choices of true worlds. $\cavg(\pi)$ is the average number of seeds chosen by $\pi$ across those worlds. 

Finally, notice that we require that the spread achieved by $\piopt$ must be $\ge Q$ in \emph{every} true world chosen and not on an average across candidate true worlds. This is motivated by practical considerations: if the policy underperforms (i.e., $f < Q$) in some worlds and the true world is among them, the marketer won't be happy! 

\eat{This is a special form of the more general minimum-cost partial cover problem in which the cost of seeding each node is different. Intuitively, this cost corresponds to the effort(eg: in terms of the percentage discount) that must be made in order to convince a user to adopt the product. A low cost thus corresponds to a loyal customer who will willingly adopt the product. In ~\cite{golovin2011adaptive}, the authors proved that the results for MINTSS can be easily generalized to the case where costs of seeding a node are non-uniform and arbitrary.}  

\eat{ 
For the adaptive version of this problem, the objective is to find a policy $\pi$ which will attain the target spread of $Q$ in every possible true world and use the minimum number of seeds. We will see how to solve the problem in its most general form (when the time horizon is bounded) in section~\ref{sec:Algorithms}. 
} 

\section{Theoretical Results}
\label{sec:Theory}
Recall, we say time horizon is unbounded if $H \ge kD$, where $k$ is the seed budget, $D$ is the length of the longest path of $G$, and $H$ is the time horizon. We consider unbounded time horizon up to Section~\ref{sec:mintss}. Bounded time horizon is addressed in subsection~\ref{sc:bounded}. Our first result is that our spread function based on node level feedback is adaptive monotone and adaptive submodular, thus affording an efficient approximation algorithm. 

\begin{theorem}
For unbounded time horizon, if the diffusion process is allowed to complete after every intervention, node level feedback is equivalent to edge level feedback w.r.t. marginal gain computation and therefore the expected spread function is adaptive submodular and adaptive monotone under node level feedback.
\label{th:th1}
\end{theorem}

\begin{proof} 
We will show that node level feedback is equivalent to edge level feedback from the perspective of marginal gain computation. In~\cite{golovin2011adaptive}, the authors show that the expected spread function under edge level feedback is adaptive monotone and adaptive submodular. The above theorem will follow from this. 
Specifically, we prove that (a) for every edge level feedback based network state, there is a corresponding state based on node level feedback, which preserves marginal gains of nodes, and (b) vice versa. 

Given edge level feedback, we clearly know which nodes are active. These are precisely nodes reachable from the seeds via live edge paths in the revealed network. In the rest of the proof, we show that for each node level feedback state, there is a corresponding edge level feedback state that preserves marginal gains. Let $S_0$ be the set of seeds chosen at time $t = 0$. Given node level feedback, we can infer the corresponding edge level feedback based network state using the following rules. 
\eat{Edge level feedback states that we know the status of each edge coming out of any active node. } 
\eat{We consider 3 possible cases to derive rules for inferring the status of the edge using node level feedback.} 
Consider an edge from  an active node $u$ to node $v$. 
\eat{According to  edge level  feedback, the status of $(u,v)$ is known in all the 3 cases.} 
Notice that the status of an edge leaving  an inactive node is unknown in either feedback model. 

\noindent 
\underline{Rule 1}: If node $u$ is active, $v$ is inactive, and there is an edge from $u$ to $v$, then infer that edge $(u,v)$ is dead. \\ 
\underline{Rule 2}: If nodes $u$ and $v$ are both active and $u$ is the only in-neighbor of $v$, then conclude that the edge $(u,v)$ is live. \\ 
\underline{Rule 3}: If nodes $u$ and $v$ are both active and $u$ has more than one in-neighbor, arbitrarily set the status of the edge $(u,v)$ to be live or dead. 

\eat{ 
Edge $(u,v)$ is dead and node $v$ is inactive. If node level feedback infers that if there is a link between the an active and inactive node, that link can be concluded to be dead. 
\item Case 2: Edge $(u,v)$ is live and hence node $v$ is active. $u$ is the only active neighbour of $v$. Node level feedback infers that in this case the link $(u,v)$ is active. 
\item Case 3: Edge $(u,v)$ is live and hence node $v$ is active. $u$ is the not only active neighbour of $v$. Under the IC model, node $v$ could have become activated due to some other active neighbour and we can't conclude whether the link $u$ and $v$ is live or not. 
\end{itemize}
} 

We now show that the way edge status is chosen to be live or dead in Rule 3 plays no role in determining the marginal gains of nodes. We make the observation that if the diffusion process is allowed to complete after each intervention, the only extra information about the network that is observed using edge level  feedback over node level feedback is the status of edges between 2 active nodes. Given that the node $u$ is active, we need to calculate the marginal gain of every other node in the network for the next intervention. Next we show that the status of edges between 2 active nodes does not matter in the marginal gain computation for any node.

For the rest of the argument, we consider the both $u$ and $v$ are active and that $v$ has multiple active in-neighbours, i.e., the case that is addressed by Rule 3. Consider an arbitrary node $w$ the marginal gain of which we need to calculate. There maybe multiple paths from $w$ to a node reachable from $w$. These paths can be classified into those which involve the edge $(u,v)$ and ones which don't. The marginal gain computation involving the latter paths is independent of the status of the edge $(u,v)$. Since the diffusion process is allowed to complete, all nodes which can be reached (in the "true" possible world) from $w$ through $(u,v)$ have already been activated. Hence paths going through $(u,v)$ do not contribute to the marginal gain for $w$. Thus, the status of the edge $(u,v)$ does not matter. 
Since $w$ is any arbitrary node, we can conclude that the marginal gain of every node remains the same under states based on both feedback models. Adaptive monotonicity and submodularity are properties of marginal gains. Since marginal gains are preserved between edge level and node level feedback, it follows that these properties carry over to our node level feedback model.
\end{proof}
\vspace{-0.25 cm}

\subsection{\im} 
\label{sec:im} 
There are four types of policies -- the greedy non-adaptive policy (abbreviated GNA), the greedy adaptive policy (GA), the optimal non-adaptive policy (ONA) and the optimal adaptive policy (OA). We use $\pi_{GA,k}$ to denote the greedy adaptive policy constrained to select $k$ seeds and $\sigma(\pi_{GA,k})$ to refer to the expected spread for this policy. While previous results bound the performance of greedy (adaptive) policies in relation to optimal adaptive policies, they do not shed light on practically implementable policies under either setting. These previous results do not quite answer the question "What do we gain in practice by going adaptive?" since both the optimal non-adaptive or optimal adaptive policies are intractable. We establish relations between two key practical (and hence implementable!) kinds of policies  -- the greedy non-adaptive policy and the greedy adaptive policy -- for both \im\ and \mintss. These relations quantify the average ``adaptivity gain'', i.e., the average benefit one can obtain by going adaptive.

We first restate Theorem 7 from~\cite{chen2013near}. This theorem gives a relation between the spreads obtained using a batch greedy adaptive policy which is constrained to select seeds in batches of size $b$ and the optimal adaptive policy.
\begin{fact}
If $\sigma(\pi_{GA,lb})$ is the average spread obtained by using a greedy batch policy with a batchsize $b$ and $\sigma(\pi_{OA,mb})$ is the spread using an optimal sequential policy (the optimal policy if we are able to select one seed per intervention) constrained to selecting a number of seeds divisible by the batchsize $b$, then
\begin{equation}
\sigma(\pi_{GA,lb}) > ( 1 - e ^ { \frac{-l}{\alpha \gamma m} } ) \sigma(\pi_{OA,mb})
\label{fact:fact1}
\end{equation}
where $\alpha$ is the multiplicative error in calculating the marginal gains. $gamma$ is a constant and equal to $(\frac{e}{e - 1}) ^ {2}$.
\end{fact}

\begin{prop}
\label{prop:p1} 
Let the horizon be unbounded. Let $\pi_{GA,n_{GA}}$ be the greedy batch policy that select $n_{GA}$ seeds overall in batches of size $b_{GA}$, and let $\pi_{OA,n_{OA}}$ be the optimal adaptive policy that selects $n_{OA}$ seeds overall in batches of size $b_{OA}$. Then 
\begin{equation}
\sigma(\pi_{GA,n_{GA}}) \geq \Bigg[1 - \exp{ \bigg( - \frac{\big\lceil \frac{n_{GA}}{b_{GA}} \big\rceil }{\alpha \gamma \big\lceil \frac{n_{OA}}{b_{OA}} \big\rceil } \bigg) }\Bigg] \sigma(\pi_{OA,n_{OA}})
\label{prop:prop1}
\end{equation}
where $\alpha \ge 1$ is the multiplicative error in calculating the marginal gains and  $\gamma = (\frac{e}{e - 1}) ^ {2}$ is a constant.
\end{prop}

\begin{proof} 
Fact~\ref{fact:fact1} gives a relation between the spreads obtained by a batch greedy adaptive policy constrained to select $lb$ seeds and the optimal adaptive policy constrained to select $mb$ seeds. Both these policies are constrained to select seeds in batches of size $b$. The relation is in terms of the number of batches used by the policies. Let $l$ and $m$ be the number of batches for the greedy and optimal policies respectively. We make the following observations. First, the two policies can be constrained to select seeds in different batchsizes, $b_{GA}$ and $b_{OA}$ respectively. Next, the number of seeds selected by the policies need not be divisible by the batchsizes. We can follow a similar proof procedure as Theorem 7 in~\cite{chen2013near} and replace $l$ by $\lceil \frac{n_{GA}}{b_{GA}} \rceil$ and $m$ by $\lceil \frac{n_{OA}}{b_{OA}} \rceil$. 
\end{proof} 

\begin{theorem}
Let $\pi_{GNA,k}$ be a greedy non-adaptive policy, $\pi_{GA,k}$ and $\pi_{OA,k}$ be the greedy and optimal adaptive policies respectively with batch-sizes equal to one i.e. the adaptive policies are sequential. All policies are constrained to select $k$ seeds. Then we have the following relations:
\begin{equation}
\sigma(\pi_{GA,k}) \geq (1 - e^{-1 / \alpha \gamma}) \sigma(\pi_{OA,k})
\label{eq:maxspread-GA}
\end{equation}
\begin{equation}
\sigma(\pi_{GNA,k}) \geq (1 - \frac{1}{e} - \epsilon)^2 \sigma(\pi_{OA,k})
\label{eq:maxspread-GNA}
\end{equation}
\end{theorem}

\begin{proof} 
Proposition~\ref{prop:p1} gives us bounds on the ratio of the spread achieved by batch-greedy adaptive policy and that achieved by the optimal adaptive policy. We set $n_{OA}$ = $k$ and $b_{GA} = b_{OA} = 1$ and obtain equation~\ref{eq:maxspread-GA} of the theorem.

Theorem 2 of~\cite{asadpour2009maximizing} states that for a submodular monotone function, there exists a non-adaptive policy which obtains $(1 - 1/e - \epsilon)$ fraction of the value of the optimal adaptive policy. In our context, this implies that the spread due to an optimal non-adaptive policy constrained to select $n_{ONA}$ seeds is within a $(1 - e^{-n_{ONA}/n_{OA}} - \epsilon)$ factor of the spread of an optimal adaptive policy selecting $n_{OA}$ seeds. More precisely,
\begin{equation}
\sigma(\pi_{ONA,n_{ONA}}) \geq (1 - e^{-n_{ONA}/n_{OA}} - \epsilon) \sigma(\pi_{OA,n_{OA}})
\label{eq:adaptivity-gap}
\end{equation}
The classical result from Nemhauser~\cite{nemhauser1978analysis} states that the greedy non-adaptive algorithm obtains a $(1 - 1/e - \epsilon)$ fraction of the value of the optimal non-adaptive algorithm, where $\epsilon$ is the additive error made in the marginal gain computation. Moreover if the greedy non-adaptive policy is constrained to select $n_{GNA}$ seeds and the optimal non-adaptive policy selects $n_{ONA}$ seeds we have the following: 
\begin{equation}
\sigma(\pi_{GNA,n_{GNA}}) \geq (1 - e^{-n_{GNA}/n_{ONA}} - \epsilon) \sigma(\pi_{ONA,n_{ONA}})
\label{eq:nemhauser}
\end{equation}
Combining equations~\ref{eq:adaptivity-gap} and~\ref{eq:nemhauser}, we obtain the following result
\begin{equation}
\sigma(\pi_{GNA,n_{GNA}}) \geq (1 - e^{-n_{GNA}/n_{OA}} - \epsilon)(1 - e^{-n_{GNA}/n_{OA}} - \epsilon) \sigma(\pi_{OA,n_{OA}})
\label{eq:gna-oa}
\end{equation}
Setting $n_{GNA}$ = $n_{GA}$ = $n_{OA} = k$, we obtain equation~\ref{eq:maxspread-GNA} of the theorem. 
\end{proof}

\textbf{Discussion:} To clarify what this theorem implies, lets assume that we can estimate the marginal gains perfectly. Let's set $\epsilon = 0$ and $\alpha = 1$. We thus obtain the following relations: 
$ \sigma(\pi_{GA,k}) \geq ( 1 - e^{-1 / \gamma} ) \sigma(\pi_{OA,k})$ and 
$\sigma(\pi_{GNA,k}) \geq ( 1 - \frac{1}{e} )^2 \sigma(\pi_{OA,k})$. These two factors are almost equal (in fact non-adaptive is slightly better) and in the case of perfect marginal estimation, there is not much gain in going adaptive. This intuition is confirmed by our experiments in section~\ref{sec:adaptiveexperiments}. 

\subsection{\mintss} 
\label{sec:mintss} 
Given that it takes the optimal adaptive policy $n_{OA}$ seeds to achieve a spread of $Q$, we seek to find the number of seeds that it will take the greedy adaptive and traditional greedy non-adaptive policy to achieve the same spread. Since the non-adaptive policy can be guaranteed to achieve the target spread only in expectation, we allow it to have a small shortfall $\beta_{ONA}$. In addition, we allow both the greedy policies to have a small shortfall against their optimal variants. We formalize these notions in the following theorem.

\begin{theorem}
Let the target spread to be achieved by the optimal adaptive policy be $Q$. Let the allowable shortfall for the optimal non-adaptive policy over the optimal adaptive policy be $\beta_{ONA}$. Let $\beta_{GA}$ and $\beta_{GNA}$ be the shortfall for the greedy adaptive and non-adaptive policies over their optimal variants. Let the number of seeds required by the four policies - OA, ONA, GA and GNA be $n_{OA}$, $n_{ONA}$, $n_{GA}$ and $n_{GNA}$. Then we have the following relations
\begin{equation}
n_{GA} \leq n_{OA}(\alpha \gamma \ln( Q / \beta_{GA} ))
\label{eq:mintss-adaptive}
\end{equation}
\begin{flalign}
& n_{GNA} \leq n_{OA}\ln\bigg(\frac{Q}{\beta_{ONA} - Q \epsilon}\bigg) \ln\bigg(\frac{Q - \beta_{ONA}}{\beta_{GNA} - \epsilon(Q - \beta_{ONA})}\bigg) & \\
& n_{GNA} \leq n_{OA}\ln\bigg(\frac{Q}{\beta_{GA} - \beta_{GNA} - Q \epsilon}\bigg) \ln\bigg(\frac{Q - \beta_{GA} + \beta_{GNA}}{\beta_{GNA} - \epsilon(Q - \beta_{GA} + \beta_{GNA})}\bigg) &
\label{eq:mintss-non-adaptive}
\end{flalign}

\end{theorem}

\begin{proof}
If in proposition~\ref{prop:p1}, we set $b_{GA}$ = $b_{OA}$ = 1 , $\sigma(\pi_{OA,n_{OA}}) = Q$ and 
$\sigma(\pi_{GA,n_{GA}}) = Q - \beta_{GA}$, after some algebraic manipulation we can obtain equation~\ref{eq:mintss-adaptive} of the theorem. 
Setting $\sigma(\pi_{ONA,n_{ONA}}) = Q - \beta_{ONA}$, $\sigma(\pi_{OA,n_{OA}}) = Q$ in equation~\ref{eq:adaptivity-gap}, we obtain the intermediate relation~\ref{eq:intermediate}.
\begin{equation}
n_{ONA} \leq n_{OA} \ln(\frac{Q}{Q - \beta_{ONA}})
\label{eq:intermediate}
\end{equation}
Setting $\sigma(\pi_{ONA,n_{ONA}}) = Q - \beta_{ONA}$ and $\sigma(\pi_{GNA,n_{GNA}}) = Q - \beta_{ONA} - \beta_{GNA}$, we obtain the following relation. 
\begin{equation}
n_{GNA} \leq n_{ONA} \ln(\frac{Q - \beta_{ONA}}{\beta_{GNA} - \epsilon(Q - \beta_{ONA})})
\label{eq:intermediate-2}
\end{equation}
We constrain the spreads for the greedy adaptive and greedy non-adaptive policies to be the same. Hence, 
$Q - \beta_{GA} = Q - \beta_{ONA} - \beta_{GNA}$. Hence $\beta_{ONA} = \beta_{GA} - \beta_{GNA}$. By combining equations~\ref{eq:intermediate} and~\ref{eq:intermediate-2} and substituting $\beta_{ONA}$ as $\beta_{GA} - \beta_{GNA}$, we obtain equation~\ref{eq:mintss-non-adaptive} of the theorem. 
\end{proof}

\textbf{Discussion:} To understand the implications of this theorem, set $\alpha = 1$, $\epsilon = 0$. Let the $\beta_{GNA} = 2$ and $\beta_{GA} = 1$, thus allowing for a shortfall of only $2$ nodes in the spread. We obtain the following relations: $n_{GA} \leq n_{OA} \gamma \ln(Q / 2)$ and $n_{GNA} \leq n_{OA} \ln(Q) \ln((Q - 1) / 2)$. 

\begin{figure}
\includegraphics[width = \columnwidth]{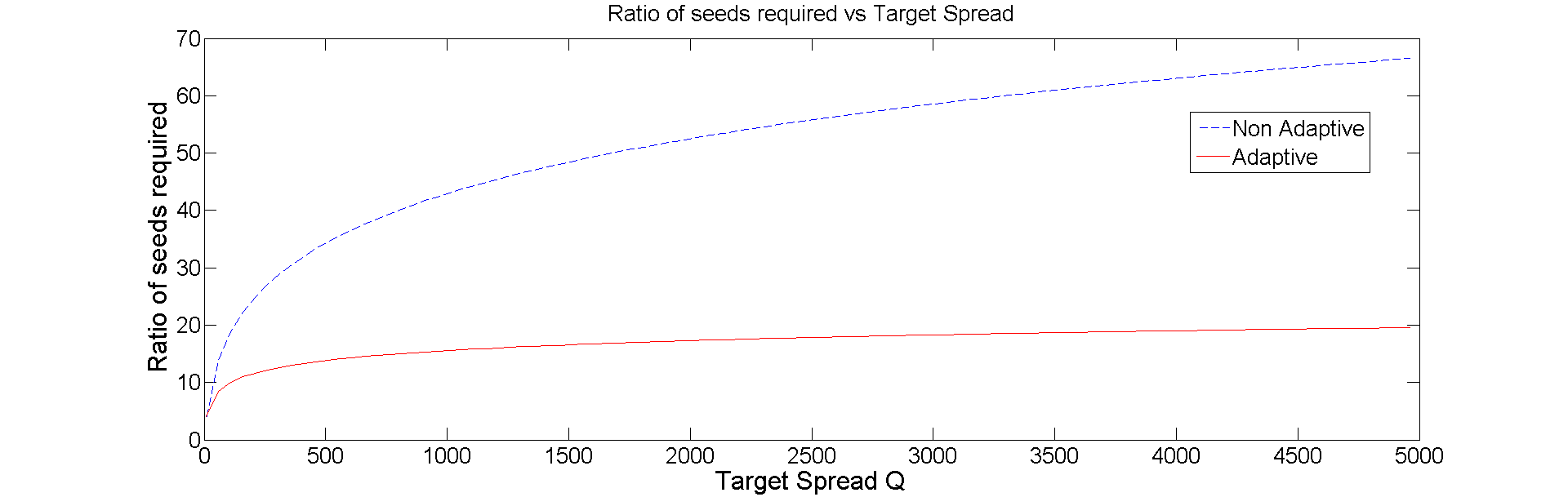}
\caption{Theoretical comparison of adaptive and non-adaptive strategies}
\label{fig:mintss-comparison}
\end{figure}
Figure~\ref{fig:mintss-comparison} shows the growth of these functions with $Q$. We can see that as $Q$ increases, the ratio $\frac{n_{GA}}{n_{OA}}$ grows much slower than $\frac{n_{GNA}}{n_{OA}}$. Hence, for the \mintss problem, there is clearly an advantage on going adaptive. This is confirmed by our experiments in section~\ref{sec:adaptiveexperiments}. 

\subsection{Bounded Time Horizon} 
\label{sc:bounded} 
In discrete diffusion models (e.g., IC), each time-step represents one hop in the graph, so the time needed for a diffusion to complete is bounded by $D$, the longest simple path in the network. 
In networks where this length is small~\cite{newman2003structure}, most diffusions complete within a short time. This is also helped by the fact that in practice, influence probabilities are small. However, if we are given a very short time horizon, the diffusion process may not complete. In this case, seed selection is forced to be based on observations of incomplete diffusions. We show that the spread function in this case is no longer adaptive submodular. 

\begin{theorem}
The spread function with the IC diffusion model is not adaptive submodular if the diffusion process after each intervention is not allowed to complete. 
\label{thm:counter}
\end{theorem}
\begin{proof}

\begin{figure}[ht]
\centering
\includegraphics[scale=0.3]{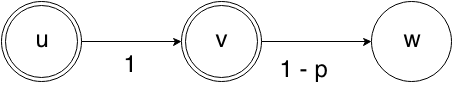}
\caption{Counterexample to show that the spread is not adaptive submodular under incomplete diffusion}
\label{fig:partial-diffusion-AS}
\end{figure}
We give a counterexample. Consider the network shown in Figure ~\ref{fig:partial-diffusion-AS} and  the true world, where the edge $(u,v)$ is live and  $(v,w)$ is dead. Let $H = 2$, $k = 2$.  Suppose at $t=0$, we choose the seed set $S = \{u\}$, so  the next intervention must be made at time $t = 1$. Based on the true world, we observe that nodes $u$ and $v$ are active at time $t = 1$. Hence we infer the edge $(u,v)$ to be live. We do not know the status of edge $(v,w)$. Even though $w$ is reachable in the network $G$, there is incomplete information in the realization revealed at $t=1$ to decide if the node $w$ is active or not, since the observed diffusion is incomplete. Thus, the expected spreads w.r.t. the realization above are as follows: $\sigma(S) = 2 + (1-p)$ and $\sigma(S \cup \{w\} ) = 3$. Let $S' = \{u,v\}$. Then $\sigma(S') = 2$ and $\sigma(S' \cup \{w\}) = 3$. This is because $w$ is one hop away from $v\in S'$ and the realization tells us that $w$ is not active. Thus, we have $\sigma(S\cup\{w\}) - \sigma(S) < \sigma(S'\cup\{w\}) - \sigma(S')$. This was to be shown. \eat{Since adaptive submodularity requires that the marginal gains should be diminishing in every realization, this example shows the spread function is not adaptive submodular in general, when the observed diffusion is forced to be incomplete. It can also be shown that we lose the node level and edge level feedback equivalence if the diffusion is not allowed to complete.} 
\end{proof}

What are our options, given that the spread under bounded time horizon is in general not adaptive submodular? Theorem 24 in ~\cite{golovin2011adaptive} shows that if a function is not adaptive submodular, no polynomial algorithm can approximate the optimal expected spread within any reasonable factor. Thus, we may continue to use adaptive greedy policy, but without any guarantees in general. In our experiments~\ref{sec:adaptiveexperiments}, we use a novel Sequential Model Based Optimization (SMBO) approach for finding a reasonably good policy when the time horizon is bounded. 

\section{Algorithms}
\label{sec:Algorithms}
To obtain a greedy adaptive policy, we need to repeatedly select nodes with the maximum marginal gain at every intervention. This implies that we need to run the greedy influence maximization algorithm to compute the marginal gain over the entire network multiple times. Fortunately, this can be done efficiently  by exploiting the recent work~\cite{Tang2014Influence} which describes a near-optimal and efficient greedy algorithm -- Two-phase Influence Maximization (TIM) for non-adaptive influence maximization. We first review TIM and describe the modifications we made to it for the adaptive case. 

\subsection{Two phase Influence Maximization}
\noindent 
{\bf Overview of TIM}: Given a budget of $k$ seeds, a network with $m$ edges and $n$ nodes and an appropriate diffusion model such as IC, TIM obtains a $(1 - 1/e - \epsilon)$ fraction of the optimal spread in the non-adaptive case, incurring a near-optimal runtime complexity of $\mathcal{O}(k+l)(n + m) log n / \epsilon^{2}$. TIM operates by generating a large number of random Reverse Reachable (RR) sets. An RR set is defined for a particular node $v$ and a possible world $W$ of the network. It consists of the set of nodes that can reach the node $v$ in the possible world $W$. Given enough number (see \cite{Tang2014Influence} for an explicit bound) of RR sets, the nodes which cover a large number of RR sets are chosen as the seed nodes: the node $u$ which appears in the maximum number of RR sets is chosen as the first seed. Once a node $u$ is selected as a seed, we remove all the RR sets containing $u$ and the next seed is the node which covers the maximum of the remaining RR sets and so on until a seed set $S$ with $k$ nodes is chosen. Tang et al.~\cite{Tang2014Influence} show that this simple strategy is enough to guarantee a $(1 - 1/e - \epsilon)$-approximation factor in near optimal time. 

\noindent 
{\bf Adaptive TIM}: 
In a greedy adaptive policy, we need to select seed nodes in every intervention. After each intervention, a certain number of nodes are influenced and become active. These already active nodes should not be selected as seeds. To ensure this, we eliminate all RR sets covered by any of these active nodes. If the number of nodes which became active is large, it brings the number of remaining RR sets below the required bound, which in turn can invalidate the theoretical guarantees of TIM, as the marginal gain of seeds selected in the next intervention may not be estimated accurately. Hence after each intervention, we need to re-generate the RR sets to effectively select seeds for the next intervention. \SV{To avoid this expensive repeated RR set generation, we instead eliminate all active nodes from the original network, by making all the incoming and outgoing edges have a zero probability, and generating the required number of RR sets for the new modified network.} This guarantees that the resulting RR sets do not contain the already active nodes. This is equivalent to running the greedy non-adaptive algorithm multiple times on  modified networks and \SV{results in retaining preserves the theoretical guarantees of TIM.} For the unbounded time horizon, the optimal policy consists of selecting one seed per intervention and letting the diffusion complete. For the IC model, the diffusion can take a maximum of $D$ time steps where $D$ is the lenght of the longest simple path in the network. 


\subsection{Sequential Model Based Optimization}
\label{subsec:bounded}
In the case of bounded time horizon (i.e.,  $T < kD$), as discussed at the end of Section~\ref{sc:bounded}, there is no straightforward strategy to find or approximate the optimal policy. The policy depends on the precise values of time horizon $T$ and properties of the network. For \im\, the two extreme cases are the non-adaptive policy  and the completely sequential greedy policy. The non-adaptive policy does not take any feedback into account and is hence suboptimal. For a sequential policy, the inter-intervention time $T/k$ will be less than $D$. Hence the completely sequential policy will result in incomplete diffusions and from ~\ref{thm:counter} will be suboptimal. A similar reasoning applies for \mintss\. For both problems, we are either forced to seed more than one node per intervention or wait for less than $D$ time-steps between interventions, or both. \SV{We split the problem of finding the optimal policy into two parts - finding the intervention times and the number of nodes to be seeded at each intervention and which nodes need to be seeded at each intervention. Using the logic in ~\ref{thm:counter}, we solve the latter problem by using the adaptive TIM algorithm described above. For the former problem, we resort to a heuristic approach since the expected spread function we need to optimize does not have any nice algebraic properties w.r.t. time. In order to find the best offline policy, we need to calculate $f_{avg}$ for each candidate policy. Calculating $f_{avg}$ across all the candidate possible worlds is expensive. Thus we need to maximize an expensive function without any nice mathematical properties. Hence we resort to a bayesian optimization technique known as sequential model based optimization (SMBO)\cite{hutter2011sequential}.} The SMBO approach narrows down on the promising configurations (in our case, policies) to optimize a certain function. It iterates between fitting models and using them to make choices about which configurations to investigate. 

We now show the above problems can be encoded for solving these problems using SMBO. Consider \im\. We have a maximum of $k$ interventions. Some of these interventions may seed multiple nodes whereas other might not seed any. There are another $k-1$ variables corresponding to the inter-intervention times. Since the number of variables is $2k - 1$, SMBO techniques will slow down as $k$ increases. It is also non-trivial to add the constraint that the sum of seeds across all interventions will add to $k$. Since this leads to an  unmanageable number of variables for large $k$, we introduce a parameter $p$ which we refer to as the policy complexity. Essentially, $p$ encodes the number of degrees of freedom a policy can have. For every $i < p$, we have a variable $s_{i}$ which is the number of nodes to be seeded at a particular intervention. We have also have a variable $t_{i}$ which encodes waiting time before making the next intervention. For example, if $p = 2$ and $s1 = 2, t1 = 5, s2 = 3, t2 = 7$ we initially seed 2 nodes, wait for 5 time-steps, then seed 3 nodes, wait for 7 time-steps before the next intervention. In the next intervention, we repeat the above procedure, until we run out of time, i.e.,  reach $T$ or get too close (within $s1$ or $s2$) to the budget of $k$ seeds. In the latter case, the last intervention just consists of using the remaining seeds. We use the same strategy to encode policies for \mintss\. In this case, however, we stop if the time reaches $T$ or if $\ge Q$ nodes become active. Since we have a manageable number of parameters, we can easily use SMBO techniques to optimize over these parameters. The objective function for the first problem is to maximize the spread. The constraint is covered by the encoding. For the second problem, the objective function is to minimize the seeds to achieve a spread of $Q$. This can be modelled by introducing penalty parameters $\lambda_{1}$ and $\lambda_{2}$. The function can be written as,
\begin{equation}
minimize\; g(x) + \lambda_{1}( Q - f(x) ) + \lambda_{2}( f(x) - Q ) 
\label{eq:SMAC-min-cost-eq}
\end{equation}
where $x$ is the parameter vector, $g(x)$ is the number of seeds, $f(x)$ is the spread, $Q$ is the target spread. The parameter $\lambda_{1}$ penalizes not achieving the target spread whereas $\lambda_{2}$ penalizes over-shooting the target spread. $\lambda_{1}$ encodes the hard constraint whereas $\lambda_{2}$ is used to direct the search. 

\section{Experiments}
\label{sec:Experiments}
\subsection{Datasets}
We run all our experiments on 3 real datasets -- the author collaboration network NetHEPT (15k nodes and 62k edges), the trust network Epinions (75k nodes and 500k edges) and Flixster. On NetHEPT and Epinions where real influence probabilities are not available, we set the probability of an edge into a node $v$ to $1 / in\-degree(v)$, following the popular approach ~\cite{chen2009efficient, wang2012scalable, chen2010scalable}. We use the Flixster network under the topic-aware independent cascade model of diffusion~\cite{barbieri2013topic} for which the authors learned the probabilities using Expectation Maximization. Their processed network has 29k nodes and 10 topics. We choose the topic which results in the maximum number of non-zero edge weights. The resulting sub-network of Flixster consists of 29k nodes and 200k edges. 

\subsection{Experimental Setup}
As mentioned earlier, we consider only the IC model of diffusion. We compare between greedy non-adaptive,  greedy sequential adaptive  and the batch-greedy adaptive policies. Since the actual true world is not known, we sample each edge in the network according to its influence probability and generate multiple  true worlds. Since we are interested in the performance of a policy on an average, we randomly generate generate $100$ true worlds and average our results across them. For either problem, the seeds selected by the non-adaptive policy is based on expected case calculations and remain the same irrespective of the true world. Only the performance of the policy is affected by the true possible world. Also note that for \mintss\, in some  true worlds the spread of the non-adaptive policy might be less than the target $Q$. The shortfall can be modelled by the factor $\beta$ introduced in Section~\ref{sec:Theory}. 

\subsection{Sequential Model Based Optimization}
We use Sequential Model-Based Optimization for General Algorithm Configuration (SMAC) ~\cite{hutter2011sequential}. SMAC is the state of the art tool used for automatic algorithm configuration for hard problems including SAT and Integer Linear Programming. Based on the performance of a configuration on certain kinds of benchmark instances characterized by  problem specific properties, SMAC creates a random forest model and uses it to choose promising candidate configurations to be evaluated. SMAC uses the model's predictive distribution to compute its expected positive improvement over the the incumbent (current  best solution).  This approach automatically trades exploitation for  exploration. SMAC can easily handle both numerical and categorical parameters. 

For our case, we need to optimize an expensive black-box function (as defined in the previous section) over $2p$ configurations where $p$ is the policy complexity. Because the function is hard to evaluate a simple brute-force grid search over the parameter space is not feasible. SMAC is implicitly able to leverage the structure present in the problem and come up with promising solutions to the problem. 

The benchmark instances consist of seeds for the random process generating 10 true worlds at a time. Hence, the evaluation of each configuration on each instance involves running the algorithm 10 times. We use a training set of $1000$ such instances and a separate test set of $50$ instances to evaluate the policies found by SMAC. We restrict the number of function evaluations SMAC can make to 500 and set the tuner timeout (the maximum time that can be spent by SMAC in building the random forest model and deciding which configuration to evaluate next) is set to 100 seconds.

\subsection{\im}
\subsubsection{Unbounded time horizon}
For \im\, we vary the number of seeds $k$ over $\{1,10,20,50,100\}$.  For the unbounded horizon, we compute the spread obtained using the greedy non-adaptive and the greedy adaptive sequential policies. We set $\epsilon = 0.1$. 

\begin{figure}[ht]
\centering
\includegraphics[scale=0.25]{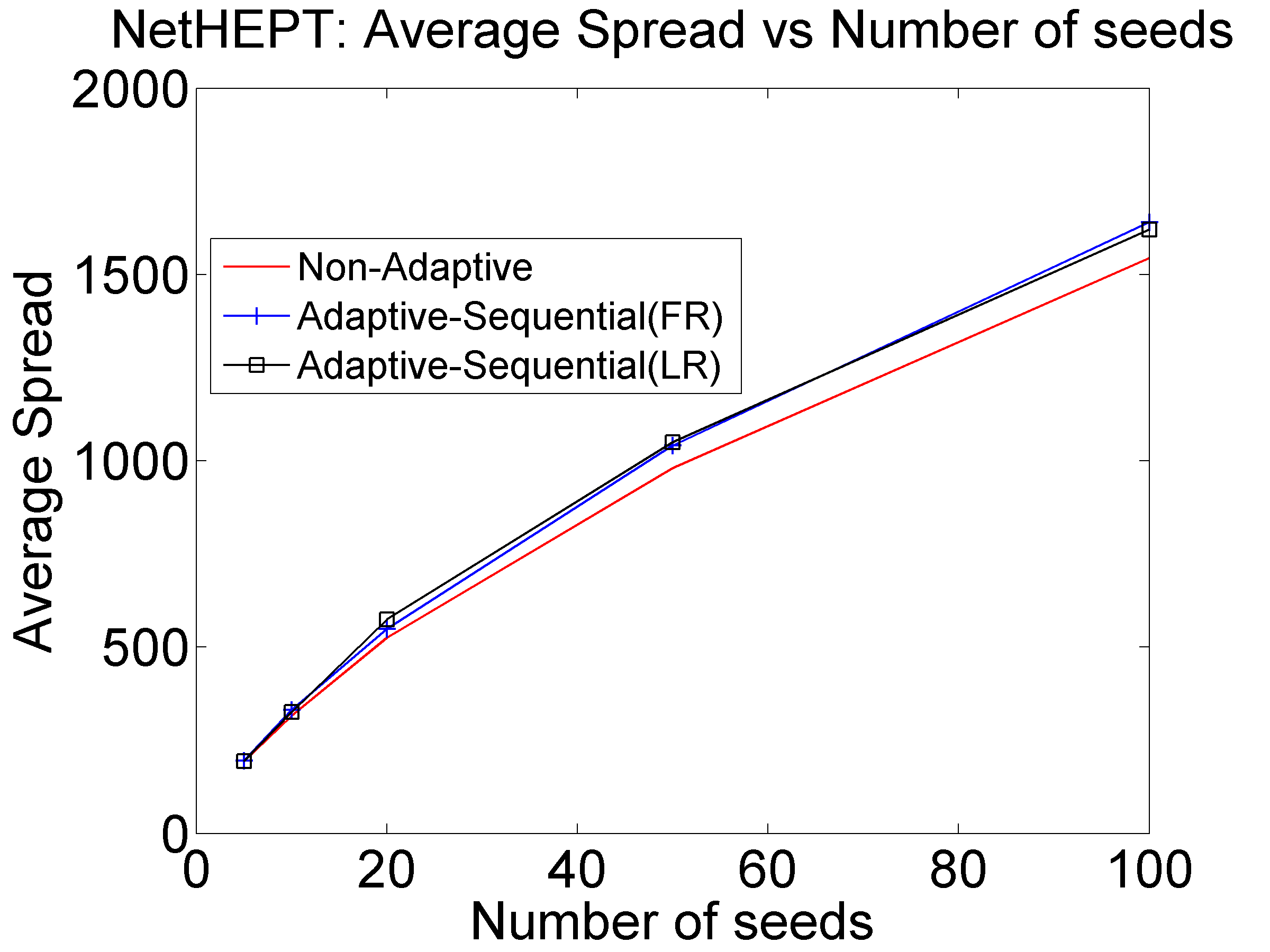}
\vspace{-0.3 cm}
\caption{NetHEPT: Average Spread vs Number of seeds}
\label{fig:NetHEPT-spread}
\end{figure}

\begin{figure}[ht]
\centering
\includegraphics[scale=0.25]{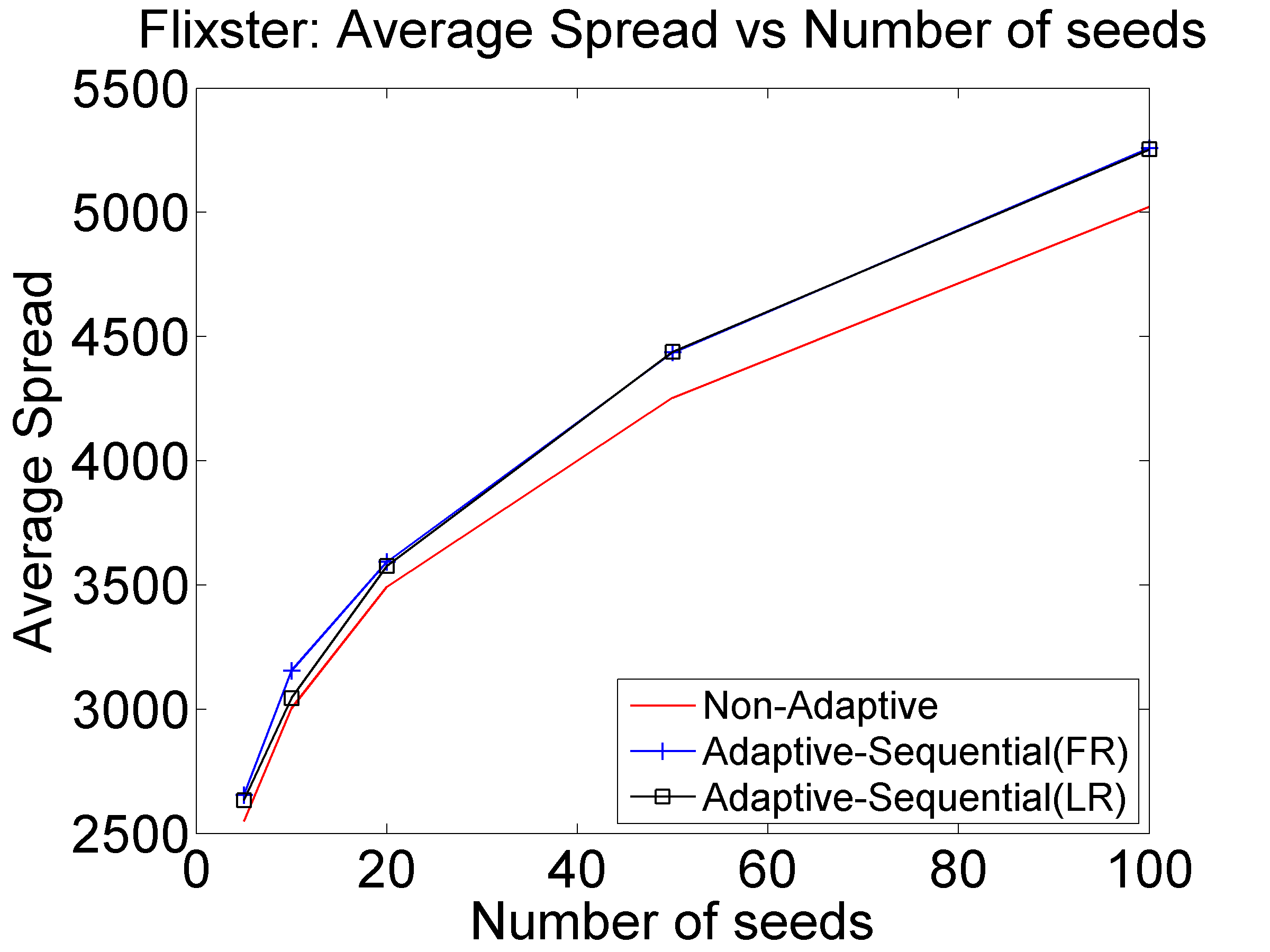}
\vspace{-0.3 cm}
\caption{Flixster: Average Spread vs Number of seeds}
\label{fig:Flixster-spread}
\end{figure}

Figures~\ref{fig:NetHEPT-spread} and~\ref{fig:Flixster-spread} show the average spread $f_{avg}$ across $100$ possible true worlds as the number of seeds is varied in the given range. We quantify the the effect of adaptivity by the ratio $\frac{f_{avg}(\pi_{GA})}{f_{avg}(\pi_{GNA})}$, which we call the average adaptivity gain. We see that the average adaptivity gain remains constant as the number of seeds are varied. We obtain similar results even with higher (100 to 500) values of $k$. This finding is consistent with the observations made in Section~\ref{sec:Theory}. 

\begin{figure}[ht]
\centering
\includegraphics[scale=0.25]{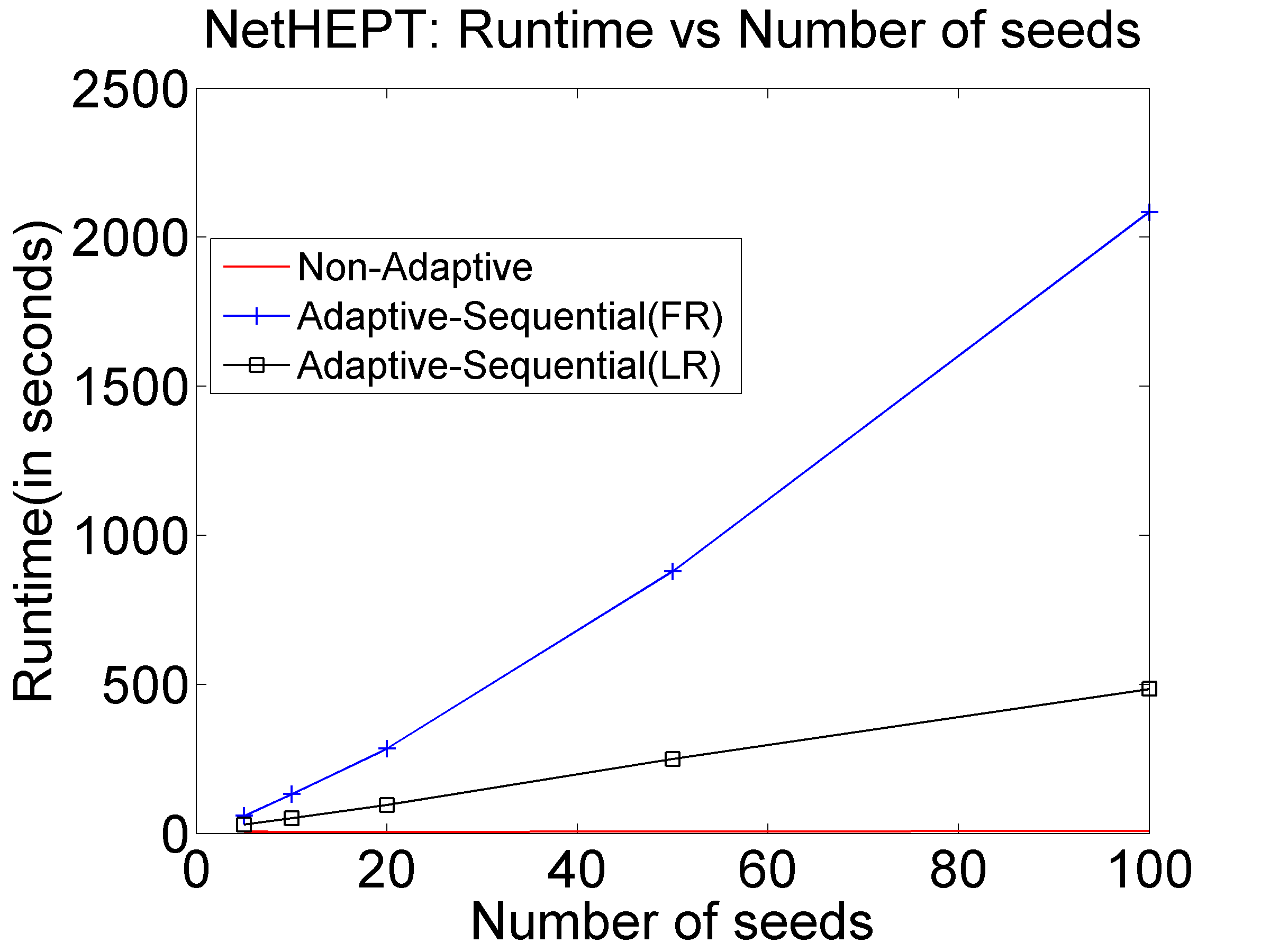}
\vspace{-0.3 cm}
\caption{NetHEPT: Runtime vs Number of seeds}
\label{fig:NetHEPT-spread-time}
\end{figure}

\begin{figure}
\centering
\includegraphics[scale=0.25]{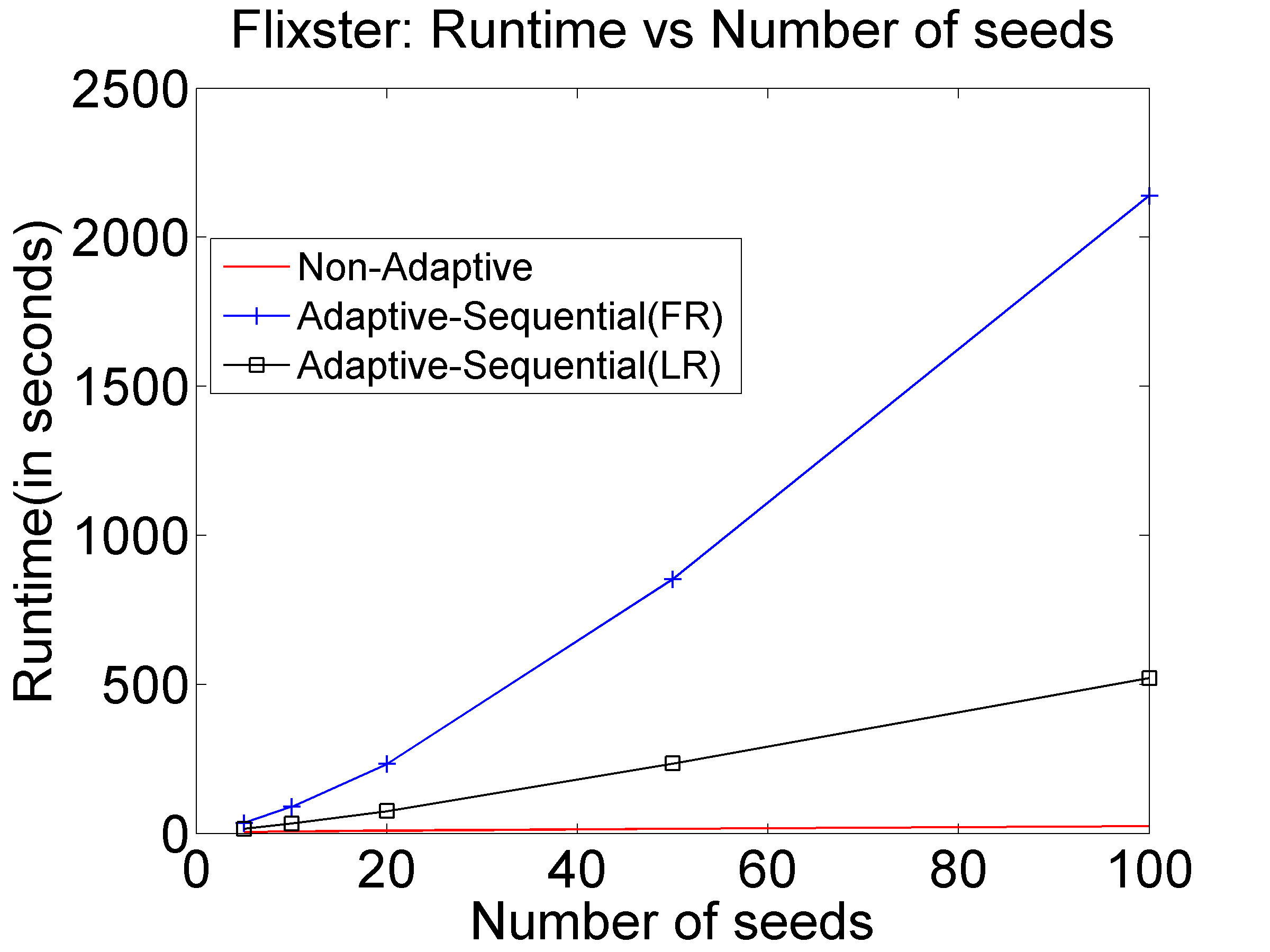}
\vspace{-0.3 cm}
\caption{Flixster: Runtime vs Number of seeds}
\label{fig:Flixster-spread-time}
\end{figure}

For the adaptive greedy sequential strategy in which we select one seed at a time, we generate RR sets for $k = 1$ and regenerate the RR sets between each pair of interventions. The run-time graphs are shown in Figures~\ref{fig:NetHEPT-spread-time} and~\ref{fig:Flixster-spread-time}. As can be seen, although this method scales linearly with the number of seeds, it is much slower than the non-adaptive case and will prohibitive for larger datasets. Instead we can generate a large number of RR sets upfront and use these sets to select seeds for the first few interventions. The RR sets are regenerated as soon as the change in the number of active nodes becomes greater than a certain threshold (the regeneration threshold $\theta$). The intuition is that if the number of active nodes has not increased much in a few interventions, the number of RR sets does not decline significantly and they still represent information about the state of the network well. We call this optimization trick lazy RR(LR) set regeneration to contrast it with the full RR(FR) set regeneration. We observe that because of submodularity, the frequency of RR set (re)generation decreases as the number of seeds (and time) increases. For our experiments, we empirically set $\theta$ equal to 10. Higher values of $\theta$ lead to lower runtimes but to a smaller average adaptivity gain. 

\begin{figure}[ht]
\centering
\includegraphics[scale=0.25]{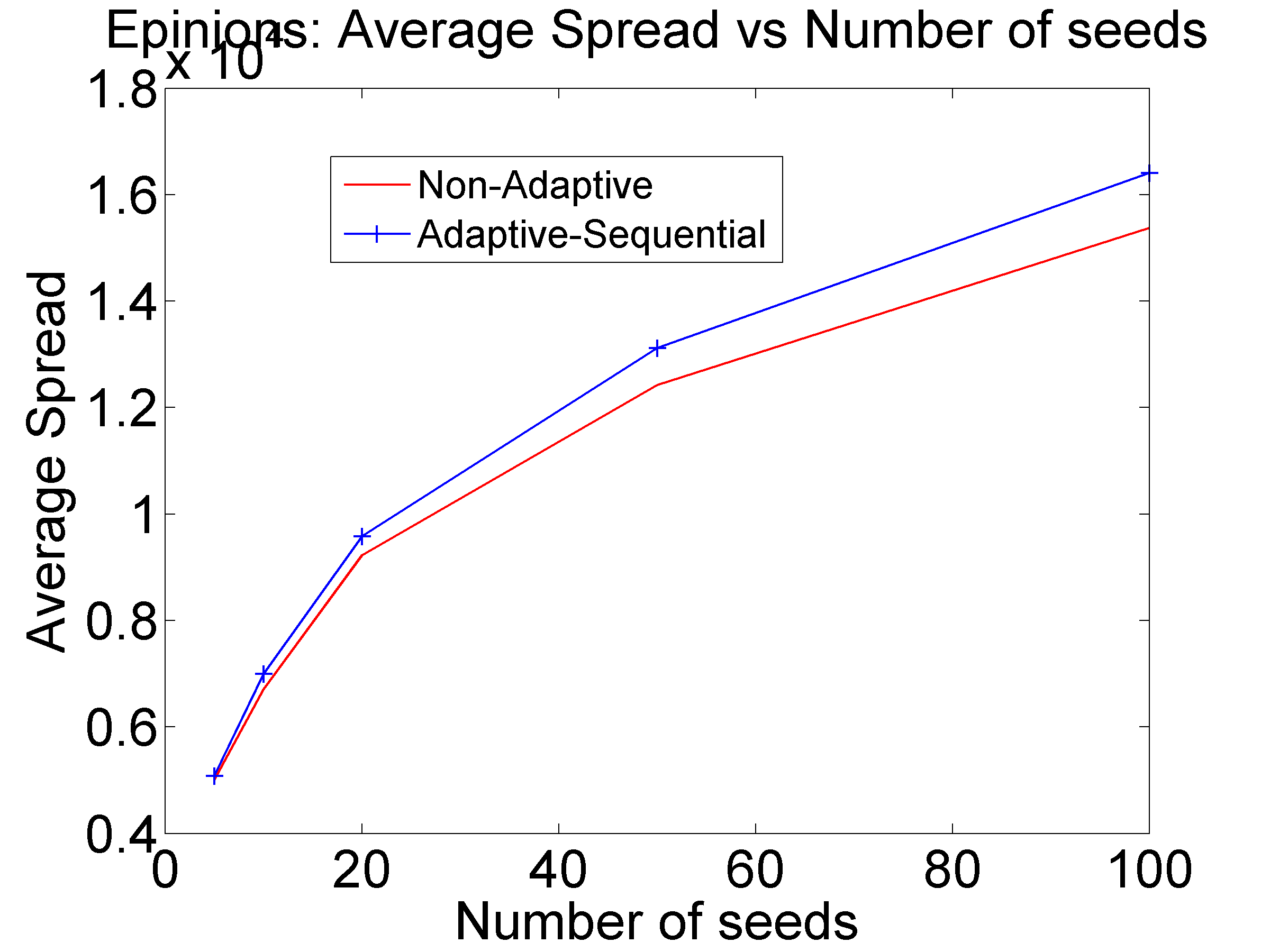}
\vspace{-0.3 cm}
\caption{Epinions: Average Spread vs Number of seeds}
\label{fig:Epinions-spread}
\end{figure}

\begin{figure}[ht]
\centering
\includegraphics[scale=0.25]{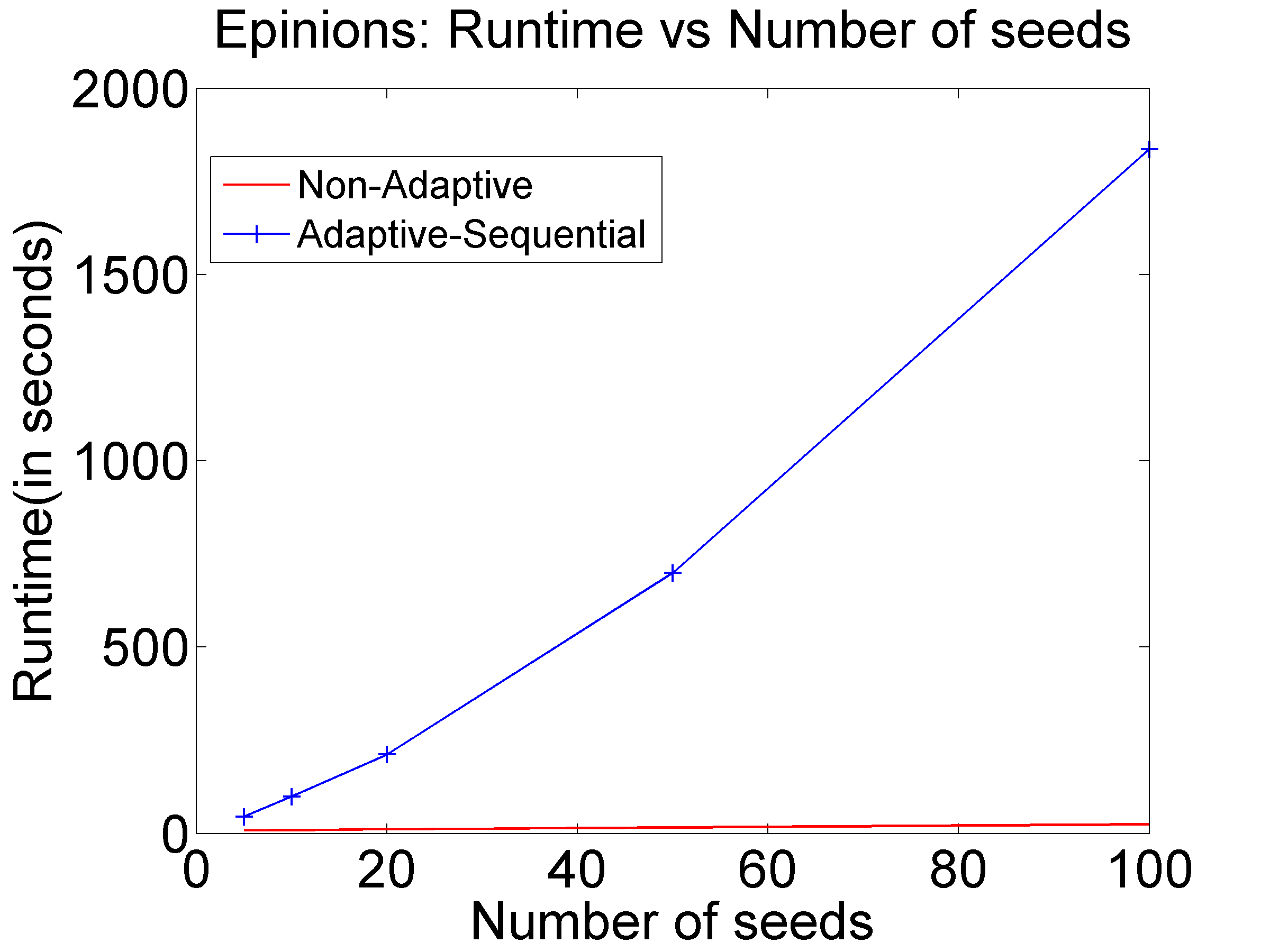}
\vspace{-0.3 cm}
\caption{Epinions: Runtime vs Number of seeds}
\label{fig:Epinions-spread-time}
\end{figure}
We use this strategy to find the spread for both NetHEPT and Flixster. As can be seen from the runtime graphs and average spread graphs, this strategy does not decrease the spread much but leads to significant computational savings. After verifying this strategy, we use it to compare the 2 policies on the larger Epinions dataset, where the same trend is observed -- see  Figures~\ref{fig:Epinions-spread} and~\ref{fig:Epinions-spread-time}. The average adaptivity gain is small even for the greedy adaptive sequential policy  in case of unbounded time horizon. 

\subsubsection{Bounded time horizon}
For the bounded time horizon, the policy will be forced to group sets of seeds together to form a batch. From Fact~\ref{fact:fact1}, we know that the average spread for such a policy will be lower and hence the average adaptivity gain will further decrease. To verify this, we conduct an experiment on the NetHEPT dataset in which we decrease the time horizon $T$ from a large value (corresponding to unbounded time horizon) to low values of the order of the length of the longest path in the network. We vary the policy complexity $p$ to be 1 or 2 in this case. We aim to find the best configuration by varying the batch-size in the range 1 to 100 and the inter-intervention time between 1 and the $D$ of the network. Since the difference between the spreads for the non-adaptive policy vs. the greedy adaptive sequential policy is so small, for the bounded time horizon, SMAC is unable to find a unique optimal policy. Different runs of SMAC yield different policies for the same number of seeds, sometimes converging to the non-adaptive policy even for reasonably large time horizons! A higher configuration time for SMAC might lead to stable results or alternatively we might need to encode the problem differently. We leave this for future work. 

\subsection{\mintss}
\subsubsection{Unbounded time horizon}
For all 3 datasets, for the unbounded time horizon, we compare the greedy non-adaptive and greedy adaptive policies with different batch sizes in the range $\{1,10,50,100\}$. Because a large number of seeds may be required to saturate a certain fraction of the network, we use the lazy RR set generation approximation explained above and set $\epsilon$ to $0.2$. Figures~\ref{fig:NetHEPT-cover},~\ref{fig:Flixster-cover} show the comparison between the non-adaptive and various adaptive greedy policies for the NetHEPT and Flixster datasets. Epinions shows a similar trend.

\begin{figure}[ht]
\centering
\includegraphics[scale=0.25]{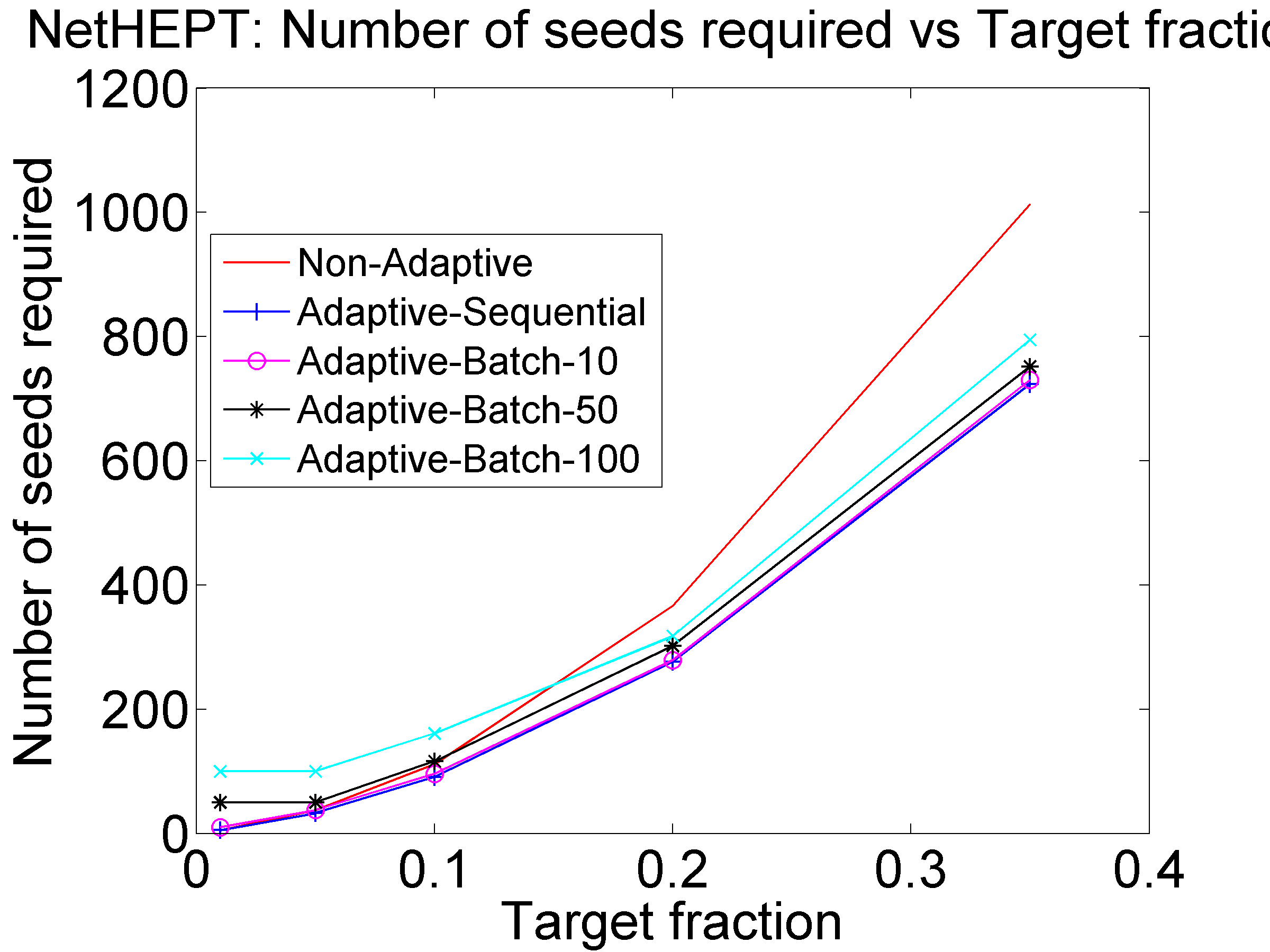}
\vspace{-0.3 cm}
\caption{NetHEPT: Number of seeds required vs Target fraction}
\label{fig:NetHEPT-cover}
\end{figure}

\begin{figure}[ht]
\centering
\includegraphics[scale=0.25]{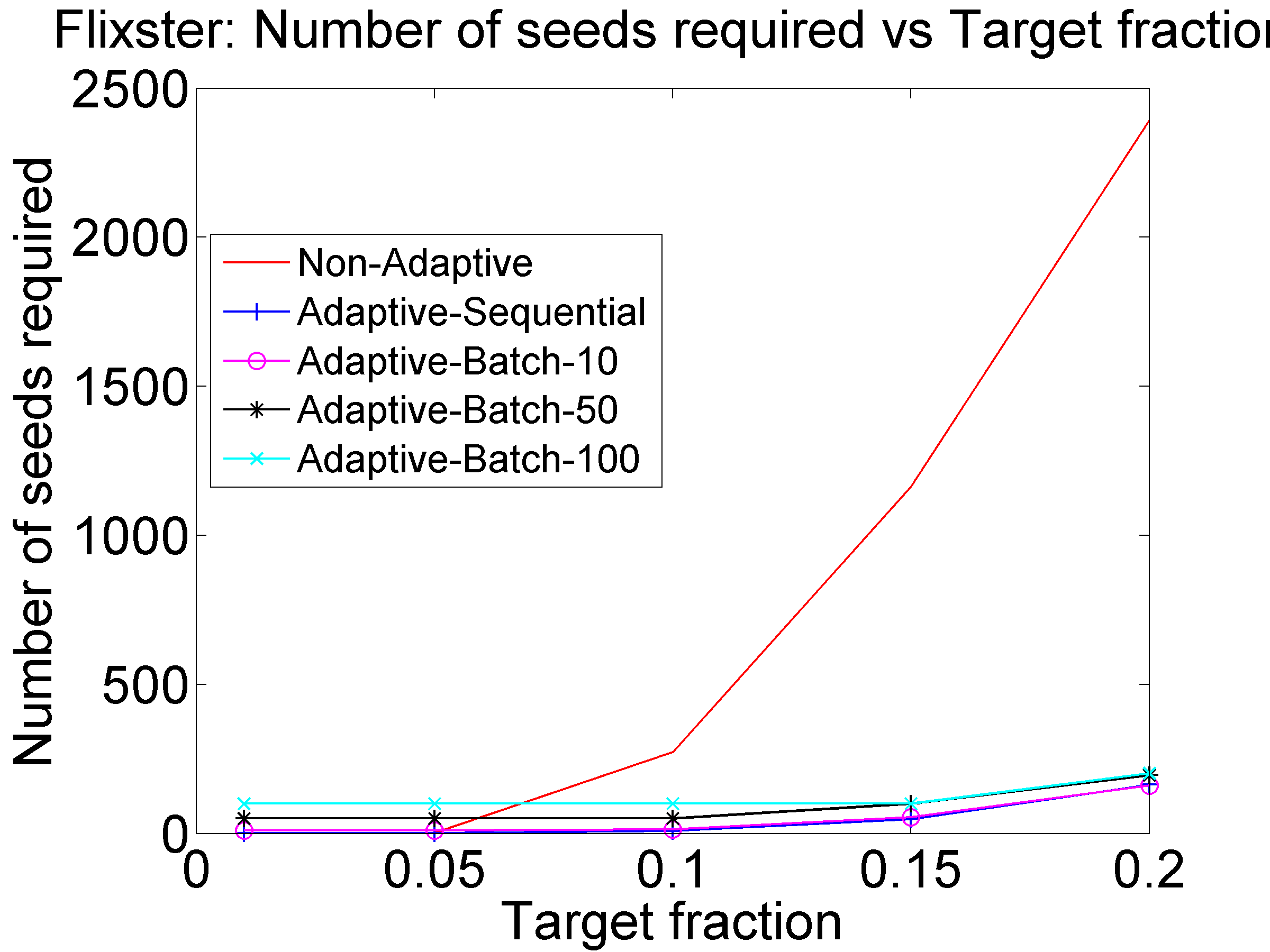}
\vspace{-0.3 cm}
\caption{Flixster: Number of seeds required vs Target fraction}
\label{fig:Flixster-cover}
\end{figure}
As can be seen, the non-adaptive policy is competitive for smaller number of target nodes. But as the target fraction increases, the adaptive policies are better able to exploit the market feedback mechanism and lead to large savings in the number of seeds. This again agrees with our theoretical results which showed that the adaptivity gain increases as the number of target nodes increases. As the size of the network increases, the estimated spread calculation in the non-adaptive case is averaged across greater number of true worlds and hence becomes less efficient. We observed that in many cases, the final true spread for the non-adaptive policy either overshoots the target spread or misses the target spread by a large amount. We conclude that adaptive policies are particularly useful if we want to influence a significant fraction of the network. 

\SV{We give some intuition for the difference in the adaptivity gains for the two problems. For adaptive policies, the rate of increase in the expected spread is fast in the beginning before the effects of submodularity take over. Hence adaptive policies require  fewer seeds than non-adaptive to reach a comparable target spread. However, once submodularity kicks in, the  additional seeds added  contribute relatively little to the spread. Hence for \mintss, where the objective is to reach a target spread with minimum seeds, the adaptivity gain is higher. However for \im, even though the adaptive policy reaches a high spread with fewer seeds, the remaining seeds in the budget don't add much more to the true spread. }


\begin{figure}[ht]
\centering
\includegraphics[scale=0.25]{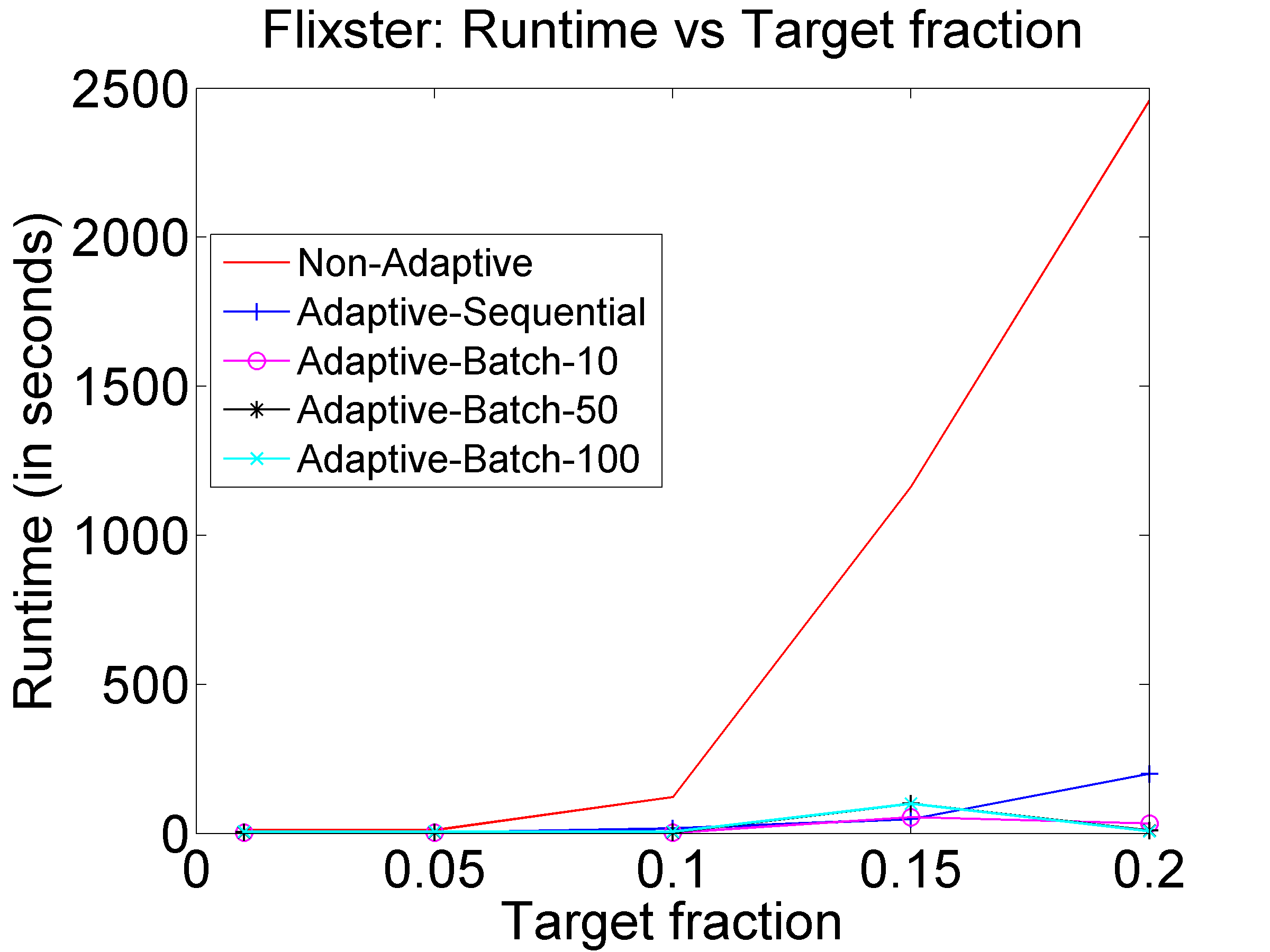}
\vspace{-0.3 cm}
\caption{Flixster: Runtime vs Target fraction}
\label{fig:Flixster-cover-time}
\end{figure}

We also plot the runtime graph for the Flixster dataset. The non-adaptive time dominates because it needs to choose a larger number of seeds. Since the batch-greedy policies select batches of seeds and consider feedback less often, they have a lower running time which decreases as the batch-size increases. Figure~\ref{fig:Flixster-cover-time} shows the runtime variation for Flixster. Results on other datasets show a similar trend. 


\subsubsection{Bounded time horizon}
\begin{table}[ht]
\begin{tabular}{ | l | c | r | r | r |}
\hline
T & 10 & 50 & 100 & 1000 \\
\hline
ShortFall ($\beta$) & 709 & 174.98 & 10.54 & 0 \\
\hline
Number of seeds & 200 & 177.33 & 171.11 & 168 \\ 
\hline
Objective function & 7290 & 1927.1 & 276.51 & 168 \\
\hline
Policy(s,t) & (100,6) & (28,8) & (20,11) & (3,12) \\
\hline
\end{tabular}
\caption{ Policies of $p = 1$ recovered by SMAC for varying time horizons($T$) for Flixster  with $Q$ = 5800 }
\label{tab:Flixster-smac-cover}
\end{table}

We now consider the important question, how good is the effect of adaptivity for a bounded time horizon for the \mintss\ problem. For this, we vary the time horizon $T$ from 10 to 1000 and the policy complexity $p$ is set to either 1 or 2. We use the Flixster dataset and fix the target fraction of nodes to 0.2. As in the previous problem, we aim to find the best configuration by varying the batch size in the range 1-100 and the inter-intervention time between 1 and the $D$ of the network. Since each configuration run involves solving \mintss\ 500 times, to save computation time we use a relatively high $\epsilon = 0.5$. We verified that similar results hold for smaller values of $\epsilon$. The optimal policy returned by SMAC is evaluated on a different set of instances (possible true worlds) averaging the results over 50 such instances. 

Table ~\ref{tab:Flixster-smac-cover} shows the results for this experiment. For both $p=1, 2$, as the time horizon increases, the shortfall goes to zero and the objective function is just the number of seeds required. We see that even for a low time horizon, SMAC is able to find a policy for which the number of seeds is close to the policy (which uses 163 seeds) for an unbounded time horizon. It is still much better than the non-adaptive version of the policy which uses a large number of seeds even for unbounded time horizon. As $T$ increases, in the policy found by SMAC, the number of seeds/interventions decreases and inter-intervention time increases. In fact, for  $T = 1000$ the $p =1$, the policy found by SMAC seeded 3 nodes per intervention and had a inter-intervention time equal to 12 (which is greater than $D$ of the graph). For extremely small $T$, the  policy found by SMAC had $100$ nodes per intervention and a very short inter-intervention time of $3$. We observe similar behaviour even for $p = 2$ and with the NetHEPT dataset as well. Note that as long as $T > D$, the non-adaptive version will require the same number of seeds it needs for the unbounded horizon case. This shows us the benefit of adaptivity even when the time horizon is severely constrained. These experiments show the effectiveness of SMAC in finding reasonably good policies for any time horizon for \mintss. 

\section{Conclusion}
\label{sec:Conclusion}
We studied adaptive influence maximization in social networks and focused on the \im\ and \mintss\ problems. We considered both the unbounded  and bounded time horizon. For unbounded horizon, we derived theoretical bounds on the performance of the greedy adaptive over greedy non-adaptive policies, thus quantifying the practical benefit of going adaptive. We studied how the adaptivity gain is affected by batch-size and number of seeds for \im\ and by target spread for \mintss. From our experiments on real networks, we conclude that while the benefit of going adaptive is modest for the \im\ problem,  adaptive policies lead to significant savings (i.e., gain) for the \mintss\ problem. For bounded time horizon, we argued that finding the optimal policy is hard and used sequential model based optimization (SMBO) techniques to find a good policy for both the problems.

Several interesting directions for future work remain. 
Extending our framework to the LT model and also to continuous time models is interesting. We believe that with continuous time models and the use of queries to infer the state of the network, adaptive influence maximization will bring the theory much closer to the practical needs of a real viral marketer.


\bibliographystyle{abbrv}
\bibliography{ref}  

\end{document}